\documentclass[11pt,conference]{IEEEtran}
\onecolumn
\usepackage{latexsym}
\usepackage{mathrsfs}
\usepackage{cite}
\usepackage{amssymb}
\usepackage{bm}
\usepackage{fullpage}
\usepackage{amsthm}
\usepackage{amsmath}
\usepackage{amssymb}
\usepackage{hyperref}
\usepackage{enumerate}
\usepackage{amsmath, amssymb}
\usepackage[dvips]{graphics}
\usepackage{graphicx}
\usepackage{psfrag}
\usepackage{enumerate}
\begin{document}

\title{Capacity and coding for the Ising Channel with Feedback \footnote{Ohad Elishco and Haim Permuter are with the Department of Electrical and 
Computer Engineering, Ben-Gurion University of the Negev, Beer-Sheva, Israel. Emails:ohadeli@post.bgu.ac.il, haimp@bgu.ac.il}}
\author{Ohad Elishco, Haim Permuter}

\maketitle

\begin{abstract}
The Ising channel, which was introduced in $1990$, is a channel with memory that models Inter-Symbol interference.
In this paper we consider the Ising channel with feedback and
find the capacity of the channel together with a capacity-achieving coding scheme.
To calculate the channel capacity, an equivalent dynamic programming (DP) problem is formulated and solved.
Using the DP solution, we establish that the feedback capacity is the expression $C=\left(\frac{2H_b(a)}{3+a} \right)\approx 0.575522$ where $a$ is
a particular root of a fourth-degree polynomial and $H_b(x)$ denotes the binary entropy function.
Simultaneously, $a=\arg \max_{0\leq x \leq 1} \left(\frac{2H_b(x)}{3+x} \right)$.
Finally, a simple, error-free, capacity-achieving coding scheme is provided together with outlining a strong connection between the
DP results and the coding scheme.
\end{abstract}


\newtheorem{question}{Question}
\newtheorem{claim}{Claim}
\newtheorem{guess}{Conjecture}
\newtheorem{definition}{Definition}
\newtheorem{fact}{Fact}
\newtheorem{assumption}{Assumption}
\newtheorem{theorem}{Theorem}
\newtheorem{lemma}{Lemma}
\newtheorem{ctheorem}{Corrected Theorem}
\newtheorem{corollary}[theorem]{Corollary}
\newtheorem{proposition}{Proposition}
\newtheorem{example}{Example}

\section{Introduction}
The Ising model originated as a problem in statistical mechanics.
It was invented by Lenz in $1920$ \cite{isingmodel1}, who gave it as a problem to his student, Ernst Ising, after whom it is named \cite{isingmodel2}.
A few years later the two dimensional Ising model was analytically defined by Onsager \cite{twodimisingchannel}.
The Ising channel, on the other hand, was introduced as an information theory problem by Berger and Bonomi in $1990$ \cite{CapacityandzeroerrorcapacityofIsingchannel}.
It has received this name due to the resemblance to the physical Ising model.

The Ising channel works as follows: at time $t$ a certain bit, $x_t$, is transmitted through the channel.
The channel output at time $t$ is denoted by $y_t$.
If $x_t=x_{t-1}$ then $y_t=x_t$ with probability $1$. If $x_t\neq x_{t-1}$ then $y_t$ is distributed Bernoulli $\left(\frac{1}{2}\right)$.

In their work on the Ising channel, Berger and Bonomi found the zero-error capacity and a numerical approximation of the capacity of the Ising channel \emph{without} feedback.
In order to find the numerical approximation, the Blahut-Arimoto Algorithm was used \cite{arimoto}, \cite{blahut}.
The capacity was found to be bounded by $0.5031\leq C \leq 0.6723$ and the zero-error capacity was found to be $0.5$ bit per channel use.
Moreover, their work contains a simple coding scheme that achieves the zero-error capacity.
This code is the basis for the capacity-achieving coding scheme in the presence of feedback presented in this paper.

We consider the Ising channel with feedback, which models a channel with Inter-Symbol Interference (ISI).
The objective is to find the channel feedback capacity explicitly, and to provide a simple, capacity-achieving coding scheme.
Finding an explicit expression for the capacity of non-trivial channels with memory, with or without feedback, is usually a very hard problem.
There are only a few cases in the literature that have been solved,
such as additive Gaussian channels with memory without feedback (``water filling solution", \cite{Shannon49waterfilling,Pinsker60}),
additive Gaussian channels with feedback where the noise is ARMA of order 1 \cite{Kim06MA},
channels with memory where the state is known both to the encoder and the decoder \cite{goldsmith96capacity,chenberger},
and the trapdoor channel with feedback\cite{Permuter06_trapdoor_submit}.
This paper adds one additional case, the Ising channel.

Towards this goal, we start from the characterization of the feedback-capacity as the normalized directed information $\frac{1}{n}I(X^n\to Y^n)$.
The directed information was introduced two decades ago by Massey \cite{causalityfeedbackanddirectedinformation} (who attributed to Marko \cite{Marko73}) as
\begin{equation}
\label{e_def_directed}
I(X^n\to Y^n) = \sum_{i=1}^n I(X^i;Y_i|Y^{i-1}).
\end{equation}
Massey \cite{causalityfeedbackanddirectedinformation} showed that the normalized maximum directed information upper bounds the capacity of channels with feedback.
Subsequently, it was shown that directed information, as defined by Massey, indeed characterizes the capacity of channels with
feedback ~\cite{ Kramer03,Kim07feedback, tat, finitestatechannels, PermuterWeissmanChenMACIT09, ShraderPermuter09CompoundIT, DaborahGoldsmith10BCfeedback}.

The capacity of the Ising channel with feedback was approximated numerically \cite{iddoandhaim} using an extension of Blahut-Arimoto algorithm for directed information.
Here, we present the explicit expression together with a simple capacity-achieving coding scheme.
The main difficulty of calculating the feedback capacity explicitly is that it is given by an optimization of an infinite-letter expression.
In order to overcome this difficulty we transform the normalized directed information optimization problem into an infinite average-reward dynamic programming problem.
The idea of using dynamic programming (DP) for computing the directed information capacity has been introduced and applied in several recent papers such
as \cite{Young,chenberger,Permuter06_trapdoor_submit, tat}.
The DP used here most resembles the trapdoor channel model \cite{Permuter06_trapdoor_submit}.
We use a DP method that is specified for the Ising channel rather then the trapdoor channel, and provide an analytical solution for the new specific DP.

It turns out that the DP not only helps in computing the feedback capacity but also provides an important information regarding a coding scheme that achieves the capacity.
Through the DP formulation and through its solution we are able to derive a simple and concrete coding scheme that achieves the feedback capacity.
The states and the actions of the dynamic programming turn out to include exact instructions of what the encoder and the decoder do to achieve the feedback-capacity.

The remainder of the paper is organized as follows:
in Section \ref{sectwo} we present some notations which are used throughout the paper, basic definitions,
and the channel model.
In Section \ref{secthree} we present the main results.
In Section \ref{s_feedback_capacity_and_dynamic_programming} we present the outline of the
method used to calculate the channel capacity. In this section we explain shortly about DP and
about the Bellman Equation, which is used in order to find the capacity.
In Section \ref{secfive} we compute the feedback capacity using a value iteration algorithm.
In Section \ref{secsix} an analytical solution to the Bellman Equation is found.
Section \ref{seceight} contains the connection between the DP results and the coding scheme.
From this connection we can derive the coding scheme explicitly.
In section \ref{secseven} we prove that the suggested coding scheme indeed achieves the capacity.
Section \ref{secnine} contains conclusions and discussion of the results.

\section{Notations, definitions and channel model}
\label{sectwo}
\subsection{Notations}

\begin{itemize}
\item Calligraphic letters, $\mathcal{X}$, denote alphabet sets, upper-case letters, $X$, denote random variables, and lower-case letters, $x$,
denote sample values.
\item Superscript, $x^t$, denotes the vector $(x_1,\dots,x_t)$.
\item The probability distribution of a random variable, $X$, is denoted by $p_X$.
We omit the subscript of the random variable when the arguments have the same letter as the random variable, e.g. $p(x|y)=p_{X|Y}(x|y)$.
\end{itemize}
The notations related to the channel are presented in Table \ref{notations}: \\
\begin{table}[h]
\caption{Frequently used notations.} \centering
\label{notations}
\begin{tabular}[h]{||c|c||}
\hline \hline
Notation & Meaning  \\
\hline \hline
$t$ & Time $(\in \mathbb{N})$ \\
\hline
$x_t$ & Channel Input at time $t$ $(\in \mathcal{X})$ \\
\hline
$s_t$ & Channel State at time $t$ $(\in \mathcal{S})$ \\
\hline
$y_t$ & Channel Output at time $t$ $(\in \mathcal{Y})$ \\
\hline \hline
\end{tabular}
\end{table}

\subsection{Definitions}
Here we present some basic definitions beginning with a definition of a finite state channel (FSC).
\begin{definition}[FSC]
\cite[ch. 4]{Gallager68} An FSC is a channel that has a finite number of possible states
and has the property: $p(y_t,s_t|x^t,s^{t-1},y^{t-1})=p(y_t,s_t|x_t,s_{t-1})$.
\end{definition}
\begin{definition}[Unifilar FSC]
\cite[ch. 4]{Gallager68} An FSC is called a unifilar FSC if there exists a time-invariant function
$f(\cdot )$ such that $s_t=f(s_{t-1},x_t,y_t)$.
\end{definition}
\begin{definition}[Connected FSC]
\cite[ch. 4]{Gallager68} An FSC is called a connected FSC if
\[
\forall s,s'\in \mathcal{S} \quad \exists \quad T_s\in \mathbb{N} \text{ and } \{p(x_t|s_{t-1})\}_{t=1}^{T_s} \text{ such that } \sum_{t=1}^{T_s}p_{S_t|S_0}(s|s')>0.
\]
In other words, for any given states $s,s'$ there exists an integer $T_s$ and an input distribution $\{p(x_t|s_{t-1})\}_{t=1}^{T_s}$,
such that the probability of the channel to reach the state $s$ from the state $s'$ is positive.
\end{definition}
\subsection{Channel model}
\label{channelmodel}
In this part, the Ising channel model is introduced.
The channel is a unifilar FSC with feedback, as depicted in Fig. \ref{basechannel}.
\begin{figure}[h!]

\begin{center}
\begin{psfrags}
    \psfragscanon
    \psfrag{a}[][][0.9]{Message}
    \psfrag{b}[][][0.9]{$x_t(m,y^{t-1})$}
    \psfrag{c}[][][0.9]{$x_t$}
    \psfrag{d}[][][0.8]{$p(y_t|x_t,s_{t-1})$}
        \psfrag{n}[][][0.8]{$s_t=f(s_{t-1},x_t,y_t)$}
    \psfrag{e}[][][0.9]{$y_t$}
    \psfrag{f}[][][0.9]{$\hat{m}(y^N)$}
    \psfrag{g}[][][0.9]{Estimated message}
    \psfrag{h}[][][0.9]{$y_t$}
    \psfrag{i}[][][0.9]{Unit}
        \psfrag{p}[][][0.9]{Delay}
    \psfrag{j}[][][0.9]{$y_{t-1}$}
    \psfrag{k}[][][0.9]{Unifilar FSC}
    \psfrag{l}[][][0.9]{Encoder}
    \psfrag{m}[][][0.9]{Decoder}
        \psfrag{o}[][][0.9]{$\hat{m}$}
        \psfrag{q}[][][0.9]{$m$}

\includegraphics[width=18cm]{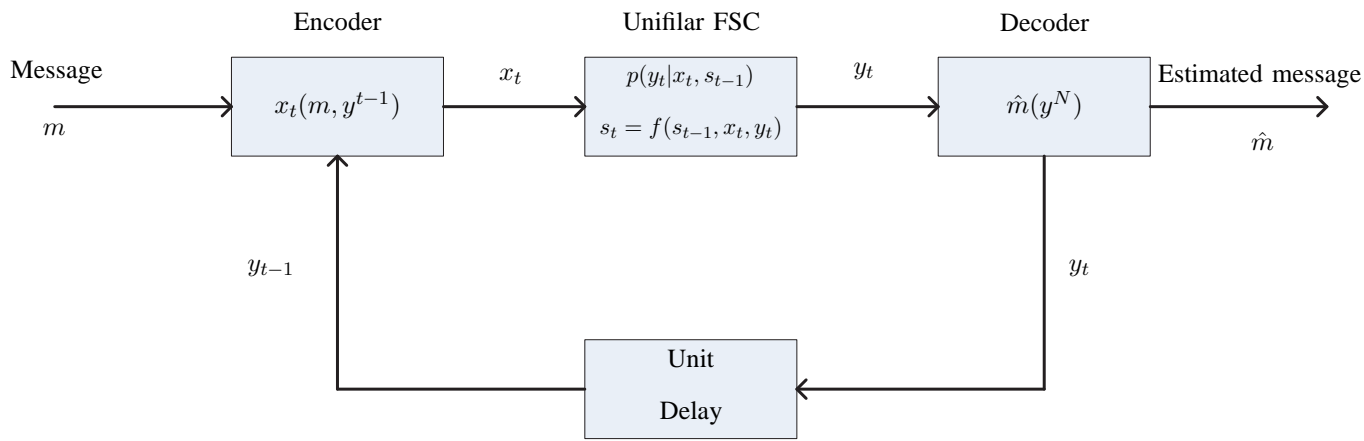}
\caption{Unifilar finite state channel with feedback of unit delay. }
\label{basechannel}
\psfragscanoff
\end{psfrags}
\end{center}
\end{figure}
As mentioned before, the sets $\mathcal{X},\mathcal{Y},\mathcal{S}$ denote the input, output, and state alphabet, respectively.
In the Ising channel model: $\mathcal{X}=\mathcal{Y}=\mathcal{S}=\{0,1\}$.

The Ising channel consists of two different topologies, as described in Fig. \ref{ising}.
The channel topologies depend on the channel state and are denoted by $Z$ and $S$.
These $Z$ and $S$ notations are compatible with the well known $Z$ and $S$ channels.
The channel topology at time $t$ is determined by $s_{t-1}\in\{0,1\}$.
\begin{figure}[h]

\begin{center}
\begin{psfrags}
    \psfragscanon
    \psfrag{d}[][][0.9]{$s_{t-1}=1$}
    \psfrag{a}[][][0.9]{$x_t$}
    \psfrag{s}[][][0.9]{$y_t$}
    \psfrag{x}[][][0.9]{$s_{t-1}=0$}
    \psfrag{z}[][][0.9]{$x_t$}
    \psfrag{c}[][][0.9]{$y_t$}

\includegraphics[width=9cm]{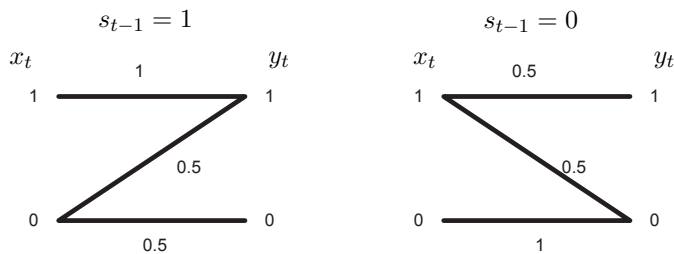}
\caption{The Ising channel model. On the left we have the $Z$ topology; and on the right we have the $S$ topology. }
\label{ising}
\psfragscanoff
\end{psfrags}
\end{center}
\end{figure}
As shown in Fig. \ref{ising}, if $s_{t-1}=1$, the channel is in the $Z$ topology; if $s_{t-1}=0$, the channel is in the $S$ topology.
The channel state at time $t$ is defined as the input to the channel at time $t$, meaning $s_t=x_{t}$.

The channel input, $x_t$, and state, $s_{t-1}$, have a crucial effect on the output, $y_t$.
If the input is identical to the previous state, i.e. $x_t=s_{t-1}$, then the output is equal to the input, $y_t=x_t$, with probability $1$;
if $x_t\neq s_{t-1}$ then $y_t$ can be either $0$ or $1$, each with probability $0.5$.
This effect is summarized in Table \ref{ta_prob_chem}. \\
\begin{table}[h]
\caption{The channel states, topologies, and inputs, together with the probability that the output at time $t$, $y_t$, is equal
to the input at time $t$, $x_t$.} \centering
\label{ta_prob_chem}
\begin{tabular}[h]{||c|c|c|c||}
\hline \hline
$s_{t-1}(=x_{t-1})$ & Topology & $x_{t}$ & $p(y_t=x_t|x_t,s_{t-1})$ \\
\hline \hline
0 & $S$ & 0 & 1 \\
\hline
0 & $S$ & 1 & $0.5$ \\
\hline
1 & $Z$ & 0 & $0.5$ \\
\hline
1 & $Z$ & 1 & 1 \\
\hline \hline
\end{tabular}
\end{table}

We assume a communication settings that includes feedback.
The feedback is with unity delay.
Hence, the transmitter (encoder) knows at time $t$ the message $m$ and the feedback samples $y^{t-1}$.
Therefore, the input of the channel, $x_t$, is a function of both the message and the feedback as shown in Fig. \ref{basechannel}.
\begin{lemma}
\label{basicprop}
The Ising channel is a connected unifilar FSC.
\end{lemma}
\begin{proof}
Lemma \ref{basicprop} is proved in three steps. At each step we show a different property of the channel.
\begin{enumerate}[(a)]
\item
The channel is an FSC channel since it has two states, $0$ and $1$.
Moreover, the output probability is a function of the input and the previous state.
Hence, it is clear that $p(y_t,s_t|x^t,s^{t-1},y^{t-1})=p(y_t,s_t|x_t,s_{t-1})$ since $s_t=x_t$ and $y_t$ depends only on $x_t$ and $s_{t-1}$.
To be more accurate, we can write $p(y_t,s_t|x^t,s^{t-1},y^{t-1})=p(y_t,s_t|x_t,s_{t-1})=p(y_t,s_t|x_t,x_{t-1})$.
\item
The channel is a unifilar FSC since (a) it is an FSC and (b) $s_t=x_t$.
Obviously, $s_t=f(s_{t-1},x_t,y_t)=x_t$ is a time-invariant function.
\item
The channel is a connected FSC due to the fact that $s_t=x_t$.
Thus, one can take $T_s=1$ and $p_{X_t|S'}(s|s')=1$, resulting in $\Pr (S_1=s|S_0=s')=1>0$.
\end{enumerate}
\end{proof}
\section{Main Results}
\label{secthree}

\begin{theorem}
\label{main_result}
\begin{enumerate}[(a)]
\item
\label{main_result_a}
The capacity of the Ising channel with feedback is
$C_f=\left(\frac{2H(a)}{3+a} \right)\approx 0.5755$ where $a\approx 0.4503$ is a specific root of
the fourth-degree polynomial $x^4-5x^3+6x^2-4x+1$.

\item
\label{main_result_b}
The capacity, $C_f$, is also equal to $\max_{0\leq z\leq 1} \left(\frac{2H(z)}{3+z} \right)$
where $a=\arg \max_{0\leq z\leq 1} \left(\frac{2H(z)}{3+z} \right)\approx 0.4503$.
\end{enumerate}
\end{theorem}

\begin{theorem}
\label{main_result_code}
There is a simple capacity-achieving coding scheme, which follows these rules:
\begin{enumerate}[(i)]
\item Assume the message is a stream of $n$ bits distributed i.i.d. with probability $0.5$.
\item Transform the message into a stream of bits where the probability of alternation between $0$ to $1$ and vice versa is $(1-a)$, where $a\approx 0.4503$.
\item Since the encoder may send some bits twice, the channel output at time $t$ does not necessarily corresponds to the $t$th message bit, therefore, we are working with two time frames.
The message time frame is denoted in $t$ while the encoder's and the decoder's time frame is denoted in $t'$.
In other words, we denote the $t$th message bit in $m_t$ where it corresponds to the $t'$th encoder's entry.
Following these rules:
\begin{enumerate}[(1)]
\item
\emph{Encoder:} At time $t'$, the encoder knows $s_{t'-1}=x_{t'-1}$, and we send the bit $m_t$ $(x_{t'}=m_t)$:
\begin{enumerate}[({1.}1)]
\item If $y_{t'}\neq s_{t'-1}$ then move to the next bit, $m_{t+1}$. This means that we send $m_t$ once.
\item If $y_{t'}=s_{t'-1}$ then $x_{t'}=x_{t'+1}=m_t$, which means that the encoder sends $m_t$ twice (at time $t'$ and $t'+1$) and then move to the next bit.
\end{enumerate}
\item
\emph{Decoder:} At time $t'$, assume the state $s_{t'-1}$ is known at the decoder, and we are to decode the bit $\hat{m}_t$:
\begin{enumerate}[({2.}1)]
\item If $y_{t'}\neq s_{t'-1}$ then $\hat{m}_t=y_{t'}$ and $s_{t'}=y_{t'}$.
\item If $y_{t'}=s_{t'-1}$ then wait for $y_{t'+1}$. $\hat{m}_t=y_{t'+1}$ and $s_{t'}=y_{t'+1}$.
\end{enumerate}
\end{enumerate}
\end{enumerate}
\end{theorem}

\section{Method Outline, Dynamic Program, and the Bellman Equation}
\label{s_feedback_capacity_and_dynamic_programming}
In order to formulate an equivalent dynamic program we use the following theorem:

\begin{theorem}
\cite{Permuter06_trapdoor_submit}
\label{t_feedback_capacity_unifilar_FSC1}
The feedback capacity of a connected unifilar FSC, when
initial state $s_0$ is known at the encoder and the decoder can be
expressed as

\begin{equation}
\label{FBcapacity}
C_{FB}=
\sup_{\{p(x_{t}|s_{t-1},y^{t-1})\}_{t\geq 1}} \liminf_{N
\rightarrow \infty} \frac{1}{N}\sum_{t=1}^{N}
I(X_t,S_{t-1};Y_t|Y^{t-1})
\end{equation}
where $\{p(x_{t}|s_{t-1},y^{t-1})\}_{t\geq 1}$ denotes the set
of all distributions such that
$p(x_t|y^{t-1},x^{t-1},s^{t-1})=p(x_t|s_{t-1},y^{t-1})$ for
$t=1,2,...$ .
\end{theorem}

Using Theorem \ref{t_feedback_capacity_unifilar_FSC1}
we can formulate the feedback capacity problem as an infinite-horizon average-reward dynamic program.
Then, using the Bellman Equation we find the optimal average reward which gives us the channel capacity.

\subsection{Dynamic Programs}

Here we introduce a formulation for average-reward dynamic programs.
Each DP problem is defined by a septuple $(\mathcal{Z}, \mathcal{U}, \mathcal{W}, F, P_z, P_w, g)$,
where all the functions considered are assumed to be measurable:

\begin{itemize}
\item [$\mathcal {Z}$]: is a Borel space which contains the \emph {states}.
\item [$F$]: is the function for which the discrete-time dynamic system evolves according to.
$z_t=F(z_{t-1},u_t,w_t), t\in \mathbb{N}^{\geq 1}$ (where each state $z_t \in \mathcal{Z}$).
\item  [$\mathcal {U}$]: is a compact subset of a Borel space, which contains the \emph {actions} $u_t$.
\item  [$\mathcal {W}$]: is a measurable space, which contains the \emph {disturbances} $w_t$.
\item  [$P_z$]: is a distribution from which the initial state, $z_0$, is drawn.
\item  [$P_w$]: is the distribution for which each disturbance, $w_t$, is drawn in accordance with $P_w(\cdot|z_{t-1},u_t)$,
which depends only on the state $z_{t-1}$ and action $u_t$.
\item  [$g$]: is a bounded reward function.
\end{itemize}
%

We consider a discrete-time dynamic system evolving according to
\begin{equation}
\label{e_dp_state_evolution}
z_t = F(z_{t-1},u_t,w_t), \quad t=1,2,3,\ldots.
\end{equation}
The \textit{history}, $h_t = (z_0, w_0, \ldots, w_{t-1})$, summarizes information available prior to the selection of the $t$th action.
The action, $u_t$, is selected by a function, $\mu_t$, which maps histories to actions.
In particular, given a policy, $\pi =\{ \mu_1, \mu_2, \ldots \}$, actions are generated according to $u_t = \mu_t(h_t)$.

The objective is to maximize the average reward, given a bounded reward function $g: \cal{Z} \times \cal{U} \rightarrow \mathbb{R}$.
The average reward for a policy $\pi $ is defined by
\begin{equation}
\label{average_reward}
\rho_{\pi} = \liminf_{N \to \infty} \frac{1}{N} \mathbb{E}_{\pi}\left\{\sum_{t=0}^{N-1}
g(Z_t, \mu_{t+1}(h_{t+1})) \right\},
\end{equation}
where the subscript $\pi$ indicates that actions are generated by the policy $\pi =(\mu_1,\mu_2,\ldots)$.
The optimal average reward is defined by $\rho^* = \sup_\pi \rho_\pi$.

\subsection{The Bellman Equation}

An alternative characterization of the optimal average reward is offered by the Bellman Equation.
This equation verifies that a given average reward is optimal.
The result presented here encapsulates the Bellman Equation and can be found in \cite{Arapos93_average_cose_survey}.

\begin{theorem}
\label{Bellman}
\cite{Arapos93_average_cose_survey} If $\rho \in \mathbb{R}$ and a bounded function
$h:\cal{Z} \mapsto \mathbb{R}$ satisfy
\begin{equation}
\label{eq:Bellman}
\rho + h(z) = \sup_{u \in {\cal U}}
\left(g(z,u) + \int P_w(dw | z,u) h(F(z,u,w))\right) \: \: \forall
z\in{\cal Z}
\end{equation}
then $\rho=\rho^*$.  Furthermore, if there is a function $\mu:{\cal
Z}\mapsto{\cal U}$ such that $\mu(z)$ attains the supremum for
each $z$ and satisfies (\ref{Bellman}), then $\rho_\pi = \rho^*$ for $\pi = (\mu_0,
\mu_1,\ldots)$ with $\mu_t(h_t) = \mu(z_{t-1})$ for each $t$.
\end{theorem}

\noindent It is convenient to define a DP operator $T$ by
\begin{equation}
\label{DPoperator}
(Th)(z) = \sup_{u \in {\cal U}} \left(g(z,u) + \int P_w(dw | z,u) h(F(z,u,w))\right),
\end{equation}
for all functions $h$.
Thus, the Bellman Equation can be written as $\rho \textbf{1} + h = Th$.
We also denote in $T_{\mu}h$ the DP operator restricted to the policy, $\mu$.

\section{Computing the feedback capacity}
\label{secfive}
%
%
\subsection{Dynamic Programming Formulation}

In this section we associate the DP problem, which was discussed in the previous section, with the Ising channel.
Using the notations previously defined, the state, $z_{t}$, would be the vector of channel state probabilities $[p(s_{t}=0|y^{t}),p(s_{t}=1|y^{t})]$.
In order to simplify notations, we consider the state $z_{t}$ to be the first component; that is, $z_{t}:= p(s_{t}=0|y^{t})$.
This comes with no loss of generality, since $p(s_{t-1}=0|y^{t-1})+p(s_{t-1}=1|y^{t-1})=1$.
Hence, the second component can be derived from the first, since the pair sums to one.
The action, $u_t$, is a $2\times 2$ stochastic matrix
\begin{equation}
u_t=\left[
\begin{array}
{cc} p(x_t=0|s_{t-1}=0) & p(x_t=1|s_{t-1}=0) \\
p(x_t=0|s_{t-1}=1) &p(x_t=1|s_{t-1}=1)
\end{array}
\right].
\end{equation}
The disturbance, $w_t$, is the channel output, $y_t$.
The DP-Ising channel association is presented in Table \ref{relations}.
\begin{table}[h]
\caption{The Ising channel model notations Vs. Dynamic Programming notations.} \centering
\label{relations}
\begin{tabular}[h]{||c|c||}
\hline \hline
Ising channel notations & Dynamic Programming notations \\
\hline \hline
$p(s_t=0|y^t)$, probability of the channel & $z_t$, the DP state \\
state to be $0$ given the output & \\
\hline
$y_t$, the channel output & $w_t$, the DP disturbance \\
\hline
$p(x_t|s_{t-1})$, channel input probability & $u_t$, the DP action \\
given the channel state at time $t-1$ & \\
\hline
Eq. (\ref {e_beta_iteration}) & $z_t=F(z_{t-1},u_{t-1},w_{t-1})$, the DP state \\
 & evolves according to a function $F$\\
\hline
$I(X_t,S_{t-1};Y_t|y^{t-1})$ & $g(z_{t-1},u_t)$, the DP reward function \\
\hline \hline
\end{tabular}
\end{table}

Note that given a policy $\pi = (\mu_1, \mu_2, \ldots)$, $p(s_t|y^t)$ is given in (\ref {e_beta_iteration})
(as shown in \cite{Permuter06_trapdoor_submit}),
where $\mathbf{1}(\cdot)$ is the indicator function.
The distribution of the disturbance, $w_t$, is
$p(w_t|z^{t-1},w^{t-1},u^t) = p(w_t | z_{t-1},u_t)$. Conditional
independence from $z^{t-2}$ and $w^{t-1}$ given $z_{t-1}$ is due
to the fact that the channel output is determined by the channel state and input.

\begin{figure*}[b]
\line(1,0){515}
\begin{center}
\begin{eqnarray}\label{e_beta_iteration}
p(s_t|y^t)
&=&\frac{\sum_{x_{t},s_{t-1}} p(s_{t-1}|y^{t-1}) u_t(s_{t-1},
x_t) p(y_{t}|s_{t-1},x_{t})\mathbf{1}(s_t =
f(s_{t-1},x_{t},y_{t}))} {\sum_{x_{t},s_{t},s_{t-1}}
p(s_{t-1}|y^{t-1}) u_t(s_{t-1},x_t)
p(y_{t}|s_{t-1},x_{t})\mathbf{1}(s_t = f(s_{t-1},x_{t},y_{t}))},
\end{eqnarray}
\end{center}
\end{figure*}

The state evolves according to $z_t=F(z_{t-1},u_t,w_t)$,
where we obtain the function $F$ explicitly using relations
from (\ref {e_beta_iteration}) as follows:
\begin{equation}
\label{states}
z_t=\left\{ \begin{array}{ll}
\frac{z_{t-1}u_t(1,1)+0.5(1-z_{t-1})u_t(2,1)}{z_{t-1}u_t(1,1)+0.5z_{t-1}u_t(1,2)+0.5(1-z_{t-1})u_t(2,1)}, \text{if } w_t=0\\
\frac{0.5(1-z_{t-1})u_t(2,1)}{0.5z_{t-1}u_t(1,2)+0.5(1-z_{t-1})u_t(2,1)+0.5(1-z_{t-1})u_t(2,2)}, \text{if } w_t=1.
\end{array} \right.
\end{equation}
These expressions can be simplified by defining
\begin{equation}
\label{22}
\begin{array}{ll}
\gamma_t:=(1-z_{t-1})u_t(2,2)\\
\delta_t:=z_{t-1}u_t(1,1),
\end{array}
\end{equation}
and, using the fact that $u_t(1,1)=1-u_t(1,2), u_t(2,2)=1-u_t(2,1)$, we have that
\begin{equation}
\label{11}
z_t=\left\{ \begin{array}{ll}
1+\frac{\delta_t -z_{t-1}}{1+\delta_t -\gamma_t},& \text{if } w_t=0\\
\frac{1-z_{t-1}-\gamma_t}{1+\gamma_t -\delta_t},& \text{if } w_t=1.
\end{array} \right.
\end{equation}

Note that $\gamma_t,\delta_t$ are functions of $z_{t-1}$.
As shown in (\ref{22}), given $z_{t-1}$, the action $u_t$ defines the pair $(\gamma_t,\delta_t)$ and vice versa.
From here on, we represent the actions in terms of $\gamma_t$ and $\delta_t$.
Since $u_t$ is a stochastic matrix, we have the constraints
$0\leq \delta_t \leq z_t$ and $0 \leq \gamma_t \leq 1-z_t$.

After finding the channel state probability, we can formulate the DP operator (\ref{DPoperator}).
We consider the reward to be $g(z_{t-1},u_t)=I(X_t,S_{t-1};Y_t|y^{t-1})$.
Note that if $g(z_{t-1},u_t)=I(X_t,S_{t-1};Y_t|y^{t-1})$ then using Theorem \ref{t_feedback_capacity_unifilar_FSC1} we have that
\begin{equation}
\rho^* =\sup_{\pi} \liminf_{N \to \infty} \frac{1}{N} \mathbb{E}_{\pi}\left\{\sum_{t=1}^{N}
I(X_t,S_{t-1};Y_t|y^{t-1})) \right\}=C_{FB}.
\end{equation}
First, we show that the reward function $I(X_t,S_{t-1};Y_t|y^{t-1})$ is indeed a function of $u_t$ and $z_{t-1}$.
To show this we note that
\begin{eqnarray}
\label{imp}
p(x_t,s_{t-1},y_t|y^{t-1})&=&p(s_{t-1}|y^{t-1})p(x_t|s_{t-1},y^{t-1})p(y_t|x_t,s_{t-1}).
\end{eqnarray}
Recall that $p(y_t|x_t,s_{t-1})$ is given by the channel model.
Thus, the reward is dependent only on $p(s_{t-1}|y^{t-1})$ and $p(x_t|s_{t-1},y^{t-1})=u_t$.
Since $p(s_{t-1}|y^{t-1})$ is given by $z_{t-1}$, we have that
\begin{eqnarray}
g(z_{t-1},u_t)&=&I(X_t,S_{t-1};Y_t|y^{t-1}) 
\end{eqnarray}
is indeed a function of $u_t$ and $z_{t-1}$.
Now we find $g(z_{t-1},u_t)$ explicitly for the Ising channel:
\begin{eqnarray}
\label{reward}
I(X_t,S_{t-1};Y_t|y^{t-1})
&=& H_b\left( Y_t|y^{t-1}\right) - H_b\left( X_t,S_{t-1}|Y_t,y^{t-1}\right)\\
&\stackrel{(a)}{=}& H_b \left( z_{t-1}u_t(1,1)+\frac{z_{t-1}u_t(1,2)}{2}+\frac{z_{t-1}u_t(2,1)}{2}\right) \\
&& \quad -(z_{t-1}u_t(1,2)\cdot 1+(1-z_{t-1})u_t(2,1)\cdot 1) \nonumber \\
&\stackrel{(b)}{=}& H_b\left(\frac{1}{2}+\frac{\delta_t-\gamma_t}{2}\right) +\delta_t+\gamma_t-1.
\end{eqnarray}
Where $H_b(\cdot)$ denotes the binary entropy function.
$(a)$ follows from Table \ref{tabone} where the conditional distribution $p(x_t,s_{t-1},y_t|y^{t-1})$ is calculated using (\ref{imp}) and
$(b)$ follows from the definition of $\delta$ and $\gamma$ given in (\ref{22}) and the fact that $u_t$ is a stochastic matrix.
Therefore, using (\ref{reward}) and (\ref{DPoperator}), we write the DP operator explicitly:
\begin{align}
\label{DP_operator}
(Th)(z)&= \sup_{u \in {\cal U}} \left(g(z,u) + \int P_w(dw | z,u) h(F(z,u,w))\right) \\
&= \sup_{u \in {\cal U}} \left(H_b\left(\frac{1}{2}+\frac{\delta_t-\gamma_t}{2}\right) +\delta_t+\gamma_t-1 + \int P_w(dw | z,u) h(F(z,u,w))\right)\\
&\stackrel{(a)}{=} \sup_{0\leq \delta \leq z,0\leq \gamma \leq 1-z} H_b \left( \frac{1}{2}+\frac{\delta-\gamma}{2}\right)
+ \delta +\gamma -1 +\frac{1+\delta-\gamma}{2}h\left( 1+\frac{\delta-z}{\delta+1-\gamma}\right) \nonumber \\
&\quad + \frac{1-\delta+\gamma}{2}h\left( \frac{1-z-\gamma}{1+\gamma-\delta}\right)
\end{align}

where $(a)$ follows from the fact that in the Ising channel case, $\int P_w(dw | z,u) h(F(z,u,w))$ takes the form $\sum_{w=0,1} p(w|z,u)h(F(z,u,w))$ and $F(z,u,w)$ is calculated explicitly using (\ref{11}).
We have formulated an equivalent dynamic program problem and found the DP operator explicitly.
The objective is to maximize the average reward $\rho_{\pi}$ over all policies $\pi$.
According to Theorem \ref{Bellman}, if we identify a scalar $\rho$ and bounded function $h$ that satisfies the Bellman Equation,
$\rho+Th(z)=h(z)$, then $\rho$ is the optimal average reward and, therefore, the channel capacity.
\begin{table}[h]
\caption{The conditional distribution $p(x_t,s_{t-1},y_t|p(s_{t-1}|y^{t-1})$ } \centering
\label{tabone}
\begin{tabular}[h]{||c|c|c|c||}
\hline \hline
$x_t$ & $s_{t-1}$ & $y_t=0$ & $y_t=1$\\
\hline \hline
0 & 0 & $p(s_{t-1}=0|y_{t-1})u_t(1,1)$ & 0 \\
\hline
0 & 1 & $0.5p(s_{t-1}=1|y_{t-1})u_t(2,1)$ & $0.5p(s_{t-1}=1|y_{t-1})u_t(2,1)$ \\
\hline
1 & 0 & $0.5p(s_{t-1}=0|y_{t-1})u_t(1,2)$ & $0.5p(s_{t-1}=0|y_{t-1})u_t(1,2)$\\
\hline
1 & 1 & 0 & $p(s_{t-1}=1|y_{t-1})u_t(2,2)$ \\
\hline \hline
\end{tabular}.
\end{table}
%
\subsection{Numerical Evaluation}

The aim of the numerical solution is to obtain some basic knowledge of the bounded function, $h$, which satisfies the Bellman Equation.
In order to do so, the Bellman equation is solved using a value iteration algorithm.
The algorithm generates a sequence of iterations according to
\begin{equation}
\label{eq:value-iteration}
J_{k+1} = T \left( J_k \right),
\end{equation}
where $T$ is the DP operator, as in (\ref{DP_operator}), and $J_0=0$.

For each $k$ and $z$, $J_k(z)$ is the maximal expected reward over $k$ periods given that the system starts in state $z$.
Since rewards are positive, $J_k(z)$ grows with $k$ for each $z$.
For each $k$, we define a differential reward function, $h_k(z) \triangleq J_k(z) - J_k(0)$.

For the numerical analysis, the interval $[0,1]$ was represented with a $1000$ points grid.
Furthermore, each interval, such as $\delta$, which is in $[0,z]$ and $\gamma$, which is in $[0,(1-z)]$, was also represented with a $1000$ points grid.
Obviously, the result has limited accuracy due to machine error in representation of real numbers.
The numerical solution after $20$ value iterations is shown in Fig. \ref{4plots}.
This figure shows the $J_{20}(z)$ function and the corresponding policies, $\gamma^*(z)$ and $\delta^*(z)$.
The policies are chosen numerically such that the equation $T_{\gamma^*,\delta^*}h(z) \geq T_{\gamma,\delta}h(z)$ holds for all $\gamma,\delta$ on the grid,
where $T_{\gamma,\delta}$ represents the DP operator restricted to the policy given by $\gamma,\delta$.
Moreover, Fig. \ref{4plots} shows the histogram of $z$, which represents the relative number of times each point has been occupied by a state $z$.
These values of $z$ have been calculated using (\ref{states}) over $250000$ iterations, where each iteration calculates the next point using (\ref{states}).
Each time a specific value of $z$ was visited the program adds to this value $\frac{1}{250000}$, which gives the relative frequency each point was visited.
\begin{figure*}[h]
\line(1,0){515}
\begin{center}
\begin{psfrags}
    \psfragscanon
    \psfrag{a}[][][0.9]{Value function on the $20^{th}$ iteration, $J_{20}$}
    \psfrag{b}[][][0.9]{Z}
    \psfrag{c}[][][0.9]{$J_{20}$}
    \psfrag{d}[][][0.9]{Histogram of Z}
    \psfrag{e}[][][0.9]{Z}
    \psfrag{f}[][][0.9]{relative frequency}
    \psfrag{g}[][][0.9]{Action parameter, $\delta$}
    \psfrag{h}[][][0.9]{Z}
    \psfrag{i}[][][0.9]{$\delta$}
        \psfrag{j}[][][0.9]{Action parameter, $\gamma$}
    \psfrag{k}[][][0.9]{Z}
    \psfrag{l}[][][0.9]{$\gamma$}
    \psfrag{m}[][][0.9]{$z_0$}
    \psfrag{n}[][][0.9]{$z_1$}
    \psfrag{o}[][][0.9]{$z_2$}
    \psfrag{p}[][][0.9]{$z_3$}

\includegraphics[width=15cm]{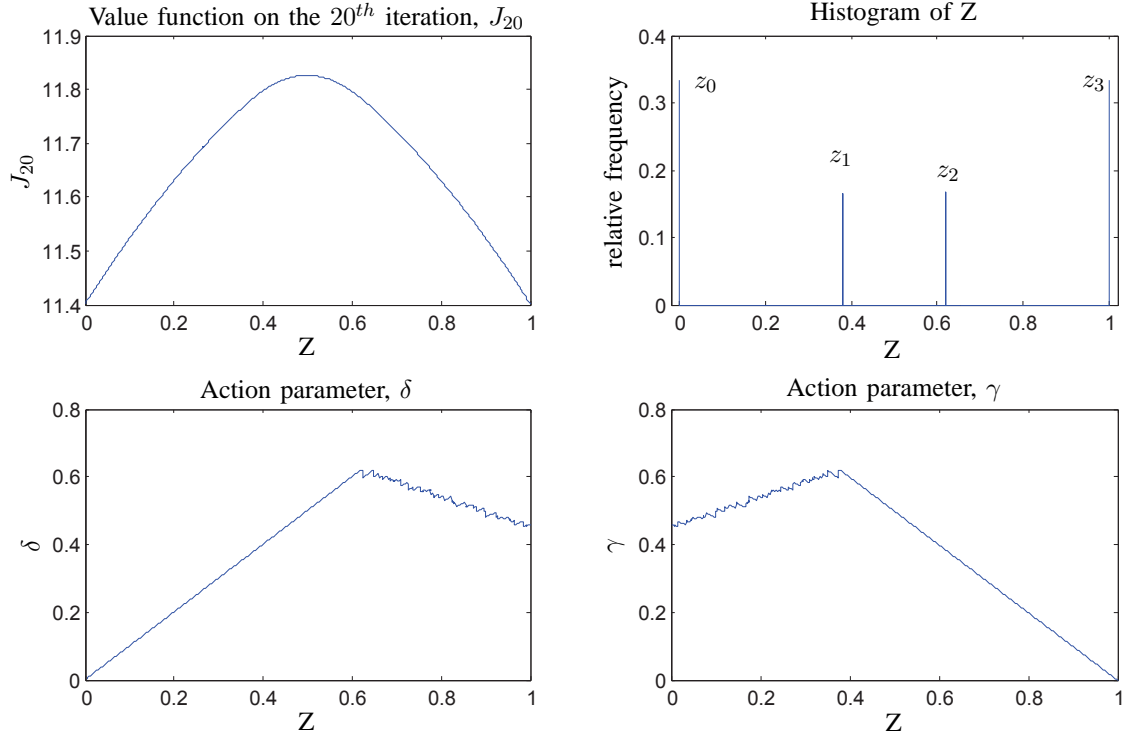}
\caption{Results from 20 value iterations.
On the top-left, the value function $J_{20}$ is illustrated.
On the top-right, the approximate relative state frequencies are shown.
At the bottom, the optimal action policies, $\delta^*$ and $\gamma^*$, are presented as
obtained after the $20^{th}$ value iteration.}
\label{4plots}
\psfragscanoff
\end{psfrags}
\end{center}
\end{figure*}

\section{Analytical Solution}
\label{secsix}
In order to ease the search after the analytical solution, some assumptions based on the numerical solution have been made.
Later on, these assumptions are proved to be correct.

\subsection{Numerical results analysis}
Here, the numerical results are explored in order to obtain further information on the possible analytical solution.
We denote the \emph{optimal} actions with $\gamma^*(z),\delta^*(z)$.
Moreover, $h^*(z)$ denotes the function $h(z)$ restricted to the policy $\gamma^*,\delta^*$.

As we can clearly see from the histogram presented in Fig. \ref{4plots}, the states, $z$, alternate between four major points.
These points, which are symmetric around $\frac{1}{2}$ and two of them are $0$ and $1$, are denoted by $z_0,z_1,z_2,z_3$, where $z_0=0$ and $z_3=1$.
In addition, the function $h^*(z)$ found numerically is assumed to be symmetric around $\frac{1}{2}$.

In order to find $z_i,\quad i=0,1,2,3$, we define $\gamma^*_0=\gamma^*(z=0)=a$.
The variables $z_i,\gamma^*_i=\gamma^*(z_i)$, and $\delta^*_i=\delta^*(z_i)$ for $i=0,1,2,3$, are presented in Table \ref{tabtwo}.
These variables are written with respect to the unknown parameter, $a=\gamma^*_0=\gamma^*(z=0)$, using (\ref{11}), and the following assumptions, which are based on Fig. \ref{4plots}:
i) we assume that $\delta=z\quad \forall z\in[z_0,z_2]$, and ii) we assume there is a symmetry relation, $\gamma^*(z)=\delta^*(1-z)$.
\begin{table}[h!]
\caption{The variables $z_i,\gamma^*_i,\delta^*_i$} \centering
\label{tabtwo}
\begin{tabular}[h]{||c|c|c||}
\hline \hline
$z_0=0$ & $\gamma^*_0=a$ & $\delta^*_0=0$\\
\hline
$z_1=\frac{1-a}{1+a}$ & $\gamma^*_1=\frac{2a}{1+a}$ & $\delta^*_1=\frac{1-a}{1+a}$ \\
\hline
$z_2=\frac{2a}{1+a}$ & $\gamma^*_2=\frac{1-a}{1+a}$ & $\delta^*_2=\frac{2a}{1+a}$ \\
\hline
$z_3=1$ & $\gamma^*_3=0$ & $\delta^*_3=a$\\
\hline \hline
\end{tabular}
\end{table}

Using the numerical solution, we assume that $a\notin \{ 0,1 \}$, which implies that $z_1,z_2\notin \{0,1\}$.
In addition, we assume that $z_1\leq z_2$ and hence $a\geq \frac{1}{3}$.
There is no loss of generality here, since otherwise we could do the same analysis and switch between $z_1$ and $z_2$.
Moreover, using (\ref{11}) and basic algebra, one can notice that the selection of $\gamma^*_i,\delta^*_i \quad i=0,1,2,3,$ as in Table.\ref{tabtwo}, guarantees that
the states, $z_t$, alternate between the points $z_i$.
For example, in order to attain the $z_i$ points, (\ref{11}) was used in the following way:
if $z_{t-1}=0$ then we have that $\gamma^*_0=a,\delta^*_0=0$ which implies
\begin{enumerate}[-]
\item if $w_t=0$ then $z_t=1+\frac{\delta^*_t-z_{t-1}}{1+\delta^*_t-\gamma^*_t}=1+\frac{0}{1-a}=1$.
\item if $w_t=1$ then $z_t=\frac{1-z_{t-1}-\gamma^*_t}{1+\gamma^*_t-\delta^*_t}=\frac{1-0-a}{1+a}=\frac{1-a}{1+a}$.
\end{enumerate}
Now, if $z_{t-1}=\frac{1-a}{1+a}$, using the symmetry, we also have the point $z_t=1-\frac{1-a}{1+a}=\frac{2a}{1+a}$.
Since $a\geq \frac{1}{3}$ we have $\frac{1-a}{1+a}\leq \frac{2a}{1+a}$.
Thus, for any specific time $t$, $z_t\in \{ z_0=0, z_1=\frac{1-a}{1+a}, z_2=\frac{2a}{1+a}, z_3=1\}$.

In addition, based on the numerical solution, we assume that $\gamma^*$ and $\delta^*$ can be approximated by straight lines.
Therefore, the following expressions can be found:
\begin{equation}
\label{gamma}
\gamma^*(z)=\left\{ \begin{array}{ll}
a+az,& \text{if } z\in [z_0,z_1]\\
1-z,& \text{if } z\in [z_1,z_3]
\end{array} \right.
\end{equation}
\\
\begin{equation}
\label{delta}
\delta^*(z)=\left\{ \begin{array}{ll}
z,& \text{if } z\in [z_0,z_2]\\
a(2-z),& \text{if } z\in [z_2,z_3].
\end{array} \right.
\end{equation}

Note that if a scalar, $\rho$, and a function, $h$, solve the Bellman Equation, so do $\rho$
and $h+c1$ for any scalar $c1$.
Hence, with no loss of generality, we can assume $h^*(\frac{1}{2})=1$.
\begin{lemma}
\label{mid_section_h}
Let $\gamma^*$ and $\delta^*$ be as in (\ref{gamma}),(\ref{delta}), then
\begin{align}
h^*(z)=H_b(z)\quad \forall z\in[z_1,z_2].
\end{align}
and $h^*(0)=h^*(1)=\rho^*$ satisfies $T_{\delta^*,\gamma^*}h^*(z)=h^*+\rho^* \quad \forall z\in[z_1,z_2]$.
\end{lemma}
\begin{proof}
Using the definition of $(T_{\delta^*,\gamma^*}h^*)(z),\delta^*$, and $\gamma^*$ we obtain
\begin{eqnarray}
h^*(z)+\rho^*&=&(T_{\delta^*,\gamma^*}h^*)(z) \nonumber \\
&=& H \left( \frac{1}{2}+\frac{z-(1-z)}{2}\right)+ z + (1-z) -1 \nonumber \\
&&\quad +\frac{1+z-(1-z)}{2}h^*\left( 1+\frac{z-z}{z+1-(1-z)}\right) + \frac{1-z+(1-z)}{2}h^*\left( \frac{1-z-(1-z)}{1+(1-z)-z}\right)  \nonumber \\
&=& H_b(z)+zh^*(1) + (1-z)h^*(0)
\end{eqnarray}
for all $z\in[z_1,z_2]$.
Using the symmetry $h^*(0)=h^*(1)$, we conclude that $(T_{\delta^*,\gamma^*}h^*)(z)=H_b(z)+h^*(0)\quad \forall z\in[z_1,z_2]$,
which implies $h^*(0)=h^*(1)=\rho^*$ and $h^*(z)=H_b(z)$.
\end{proof}
\begin{lemma}
\label{end-section-h}
Let $\gamma^*$ and $\delta^*$ be as in (\ref{gamma}),(\ref{delta}), then
\begin{align}
h^*(z)&= \left( \frac{1}{1-a} \right) H\left( \frac{2a+(1-a)z}{2} \right) -z+\frac{az-4a-z}{2(1-a)}\rho^* \nonumber \\
&\quad + \frac{2a+(1-a)z}{2(1-a)}H\left( \frac{2a}{a(2-z)+z}\right)\quad ,\forall z\in[z_2,z_3].
\end{align}
and $h^*(0)=h^*(1)=\rho^*$ satisfies $T_{\delta^*,\gamma^*}h^*(z)=h^*+\rho^* \quad \forall z\in[z_2,z_3]$.
\end{lemma}
\begin{proof}
Using the policy of $\gamma^*,\delta^*$ as in (\ref{gamma}),(\ref{delta}), one can write $(T_{\delta^*,\gamma^*}h^*)(z)\quad \forall z\in[z_2,z_3]$:
\begin{align}
\label{oneone}
h^*(z)+\rho^*&=(T_{\delta^*,\gamma^*}h^*)(z) \nonumber \\
&= H\left( \frac{2a+(1-a)z}{2}\right) +2a-(1+a)z+ \frac{2a+(1-a)z}{2}h^*\left( \frac{4a-2az}{2a+(1-a)z}\right) \nonumber \\
&\quad +\frac{2-2a+(a-1)z}{2}h^*(0).
\end{align}
Note that the argument in the function $h^*$ in (\ref{oneone}), which we denote here as
$l(z):=\frac{4a-2az}{2a+(1-a)z}$, is in $[z_2,z_3]$ for $z\in [z_2,z_3]$.
Hence, we can apply (\ref{oneone}) twice. By using simple algebra we obtain
\begin{align}
\label{h^*(z)}
h^*(z)&= \left( \frac{1}{1-a} \right) H\left( \frac{2a+(1-a)z}{2} \right) -z+\frac{az-4a-z}{2(1-a)}\rho^* \nonumber \\
&\quad + \frac{2a+(1-a)z}{2(1-a)}H\left( \frac{2a}{a(2-z)+z}\right)\quad ,\forall z\in[z_2,z_3].
\end{align}
\end{proof}
%
\begin{figure*}[t]
\begin{center}
\begin{psfrags}
    \psfragscanon
    \psfrag{a}[][][0.9]{The function $h^*(z)$}
    \psfrag{b}[][][0.9]{z}
    \psfrag{c}[][][0.9]{$h^*(z)$}
    \psfrag{d}[][][0.9]{Histogram of Z}
    \psfrag{e}[][][0.9]{z}
    \psfrag{f}[][][0.9]{relative frequency}
    \psfrag{g}[][][0.9]{Action parameter, $\delta^*$}
    \psfrag{h}[][][0.9]{z}
    \psfrag{i}[][][0.9]{$\delta^*$}
        \psfrag{j}[][][0.9]{Action parameter, $\gamma^*$}
    \psfrag{k}[][][0.9]{z}
    \psfrag{l}[][][0.9]{$\gamma^*$}

\includegraphics[width=15cm]{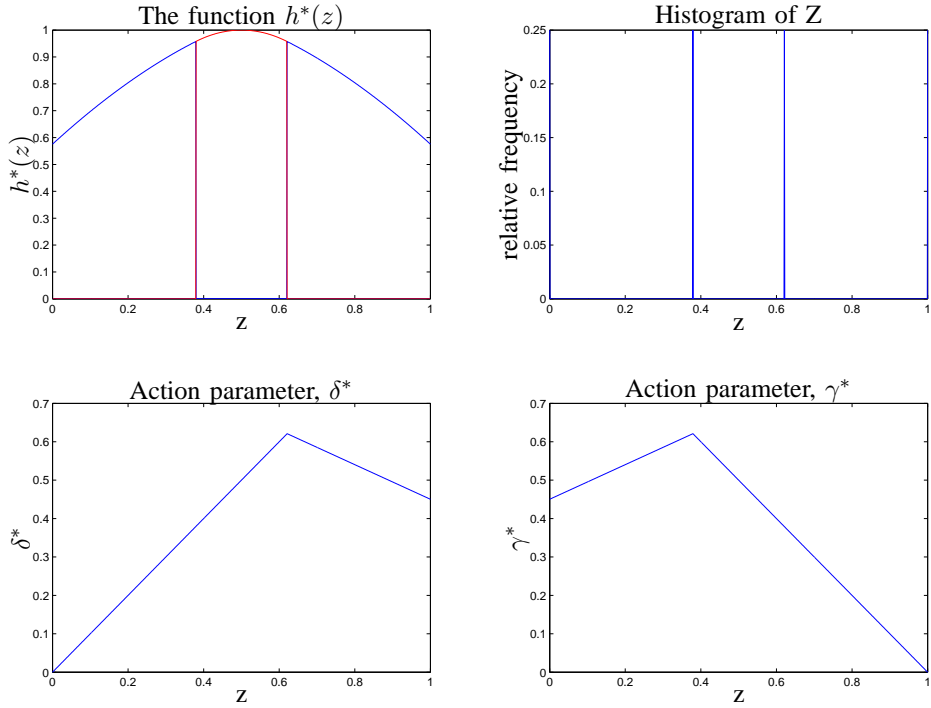}
\caption{The results using the numerical analysis. On the top-left, the value function $h^*(z)$ is shown.
On the top-right, the assumed relative state frequencies are shown (the figure presents only the states).
At the bottom, the optimal action policies $\delta^*$ and $\gamma^*$ are illustrated.
As can be seen, these match perfectly the results obtained numerically after
the $20^{th}$ value iteration.}
\label{matlab1}
\psfragscanoff
\end{psfrags}
\end{center}
\line(1,0){515}
\end{figure*}
Using the symmetry relation, one can derive $h^*(z)$ for $z\in[z_0,z_1]$.
If we set $z=1$ in (\ref{h^*(z)}) and take $\rho^*=h^*(1)$ we obtain that $\rho^*=\frac{2H(a)}{a+3}$.
This expression is examined later on.

In order to find the variable $a$ we make sure that our $\gamma^*,\delta^*$ are indeed the arguments that maximize $Th^*(z)$.
To show this, we differentiate the expression $T_{\gamma,\delta}h^*(z)$ with respect to $\delta$ and set the result
to $0$. Recall that $T_{\gamma,\delta}$ is the DP operator restricted to the policy $\gamma,\delta$, hence it is the DP operator without the supremum.
Using basic algebra and setting $\gamma^*,\delta^*$ as in (\ref{gamma}), (\ref{delta}) with $a$ as a variable,
we obtain
\begin{equation}
\frac{\partial T_{\delta,\gamma}h^*(z)}{\partial \delta}=\frac{\rho^*}{a-1}+\frac{\log_2(a)}{2(a-1)}
\end{equation}
with the notation of $0 \cdot \log (0) = 0$, since $\lim_{t\rightarrow 0} t\cdot \log(t)=0$.
Replacing $\rho^*$ with $\frac{2H(a)}{a+3}$ and setting the result to zero
we find the variable $a$ to be the root of a fourth-degree polynomial, $x^4-5x^3+6x^2-4x+1$, where, using the Ferrari
Formula, it can be found explicitly.
We find that this polynomial has two imaginary roots, one root which is grater then 1 ($\approx 3.63$), and one root in $[0,1]$,
which is roughly $0.4503$.
Hence, the only suitable value is $a\approx 0.4503$.
Calculating $\rho^*=\frac{2H(a)}{a+3}$ we find that $\rho^*\approx 0.575522$.

In Fig. \ref{matlab1}, the results using the numerical analysis are presented.
The upper pictures present the function
\begin{equation}
\label{thefunction}
h^*(z)=\left\{ \begin{array}{ll}
H_b(z),& \text{if } z\in [\frac{1-a}{1+a},\frac{2a}{1+a}]\\
\left(\frac{1}{1-a}\right) H\left(\frac{2a+(1-a)z}{2}\right)-z+\frac{az-4a-z}{2(1-a)}\rho \\
+\frac{2a+(1-a)z}{2(1-a)} H \left( \frac{2a}{a(2-z)+z} \right),& \text{if } z\in [\frac{2a}{1+a},1]
\end{array} \right.
\yesnumber
\end{equation}
as found using the numerical evaluation and the histogram of $z$,
where for $z\in [0,\frac{1-a}{1+a}]$ we derive $h^*(z)$ using the symmetry of $h^*(z)$ with respect to $\frac{1}{2}$.
The histogram shows all the values which were occupied by the state $z_t$ for some $t$;
the relative frequency shows how many times a value has been occupied.
The bottom line presents the actions parameters $\delta^*(z)$ and $\gamma^*(z)$ as obtained numerically.
Fig. \ref{matlab2} shows the function $J_{20}$, which was attained after $20$ value iterations together with the function $h^*(z)$ obtained in the numerical analysis.

\begin{figure}[h!]
\begin{center}
\begin{psfrags}
    \psfragscanon
    \psfrag{a}[][][0.9]{The functions $h(z)$ and $J_{20}$}
    \psfrag{b}[][][0.9]{z}
    \psfrag{c}[][][0.9]{$h(z)$ and $J_{20}$}

\includegraphics[width=10cm]{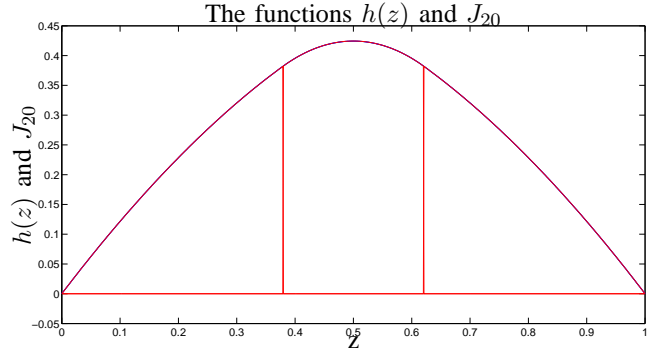}
\caption{The function $J_{20}$ as found numerically and the function $h^*(z)$ found using the numerical analysis
on the same plot. The two results match perfectly.}
\label{matlab2}
\psfragscanoff
\end{psfrags}
\end{center}
\end{figure}

\subsection{Analytical Solution Verifications}
In this section we verify that the function $h^*(z)$, as found in (\ref{thefunction}), and $\rho^*=\frac{2H(a)}{3+a}$, indeed satisfy the Bellman Equation, $Th^*(z)=h^*(z)+\rho^*$.
Furthermore, we show that the selection of $\delta^*,\gamma^*$, as in (\ref{delta}), (\ref{gamma}) maximizes $Th^*(z)$.
Namely, we prove Theorem \ref{main_result}.

We begin with proving the following lemma:
\begin{lemma}
\label{mean1}
The policies $\gamma^*,\delta^*$ which are defined in (\ref{gamma}), (\ref{delta}):
\begin{equation*}
\gamma^*(z)=\left\{ \begin{array}{ll}
a+az,& \text{if } z\in [z_0,z_1]\\
1-z,& \text{if } z\in [z_1,z_3]
\end{array} \right.
\end{equation*}
\\
\begin{equation*}
\delta^*(z)=\left\{ \begin{array}{ll}
z,& \text{if } z\in [z_0,z_2]\\
a(2-z),& \text{if } z\in [z_2,z_3]
\end{array} \right.
\end{equation*}
maximize $Th^*(z)$, i.e. $T_{\delta^*,\gamma^*}h^*(z)=Th^*(z)$, where $a\approx 0.4503$ is a specific root of the fourth-degree polynomial $x^4-5x^3+6x^2-4x+1$.
\end{lemma}

In order to prove Lemma \ref{mean1} we need two main lemmas, which use the following notation:
we denote the expression $g(\delta,\gamma)=T_{\delta,\gamma}h^*(z)$ (which is $Th^*(z)$ without the supremum, restricted to the policy $(\delta,\gamma)$) as follows:
\begin{align}
g(\delta,\gamma)&=H\left(\frac{1}{2}+\frac{\delta-\gamma}{2}\right) +\delta +\gamma -1
+\frac{1+\delta-\gamma}{2}h^*\left( 1+\frac{\delta-z}{\delta+1-\gamma}\right) \nonumber \\
& \quad +\frac{1-\delta+\gamma}{2} h^*\left(\frac{1-z-\gamma}{1+\gamma-\delta}\right).
\end{align}

The first main lemma shows that the function $g(\delta,\gamma)$ is concave in $(\delta,\gamma)$.
To prove this lemma we first show that the concatenation of any finite collection of continuous concave functions,
$\{f_i:[\alpha_{i-1},\alpha_{i}]\rightarrow \mathbb{R}\} \quad i=1,2,\cdots,n$,
where $\alpha_{i-1}\leq \alpha_{i}$ and $\alpha_i\in \mathbb{R}$,
which each have the same derivative at the concatenation points $\left(f'_{i-}(\alpha_i)=f'_{(i+1)+}(\alpha_i)\right)$,
is a continuous concave function.
It is sufficient to show that the concatenation of two continuous, concave functions with the same derivative at the concatenation point is continuous and concave.
The proof for any finite collection of such functions results using induction.

\begin{lemma}
\label{lemma_ohad}
Let $f:[\alpha,\beta]\rightarrow \mathbb{R}$, $g:[\beta,\gamma]\rightarrow \mathbb{R}$ be two continuous,
concave functions where $f(\beta)=g(\beta)$,
$f'_-(\beta)=g'_+(\beta)$, where $f'_-(\beta)$ denotes the left derivative of $f(x)$ at
$\beta$ and $g'_+(\beta)$ denotes the right derivative of $g(x)$ at $\beta$.
The function obtained by concatenating $f(x)$ and $g(x)$ defined by
\begin{equation}
\eta(x)=\left\{ \begin{array}{ll}
f(x),& \text{if } x\in [\alpha,\beta]\\
g(x),& \text{if } x\in [\beta,\gamma]
\end{array} \right.
\end{equation}
is continuous and concave.
\end{lemma}
The proof of Lemma \ref{lemma_ohad} is in the appendix.

We now conclude that the function $h^*(z)$ is concave in $z$.
\begin{corollary}
\label{cor_ohad}
The function $h^*(z)$ as given in (\ref{thefunction}) is continuous and concave for all $z\in [0,1]$.
\end{corollary}

\begin{proof}
It is well known that the binary entropy function, $H_b(z)$, is concave in $z$. Thus, the function $h^*(z)$ for $z\in [z_1,z_2]$ is concave.
In order to show that $h^*(z)$ for $z\in [z_0,z_1]$ and for $z\in [z_2,z_3]$ is concave we first notice that
$\left(\frac{1}{1-a}\right) H\left(\frac{2a+(1-a)z}{2}\right)$ is concave since it is a composition of the binary entropy function and a linear, non-decreasing function of $z$.
Second, the expression  $-z+\frac{az-4a-z}{2(1-a)}\rho$ is also concave in $z$ since it is linear in $z$, and third,
the expression $\frac{2a+(1-a)z}{2(1-a)} H \left( \frac{2a}{a(2-z)+z} \right)$ is concave using the perspective property of concave functions.
Hence, the sum of the three expression is also concave, which implies that $h^*(z)$ is concave in $z$ for $z\in [z_0,z_1]$, $z\in [z_1,z_2]$, and for $z\in [z_2,z_3]$.
It is easy to verify that the function $h^*(z)$ is a concatenation of three functions that satisfies the conditions in Lemma \ref{lemma_ohad}.
Thus, we conclude that $h^*(z)$ is continuous and concave for all $z\in [0,1]$.
\end{proof}

Using the previous corollary we obtain the following:
\begin{lemma}
\label{lemma_concave}
Let $h(z)$ be a concave function. The expression given by
\begin{align*}
H\left(\frac{1}{2}+\frac{\delta-\gamma}{2}\right) +\delta +\gamma -1
+\frac{1+\delta-\gamma}{2}h\left( 1+\frac{\delta-z}{\delta+1-\gamma}\right)
+\frac{1-\delta+\gamma}{2} h\left(\frac{1-z-\gamma}{1+\gamma-\delta}\right)
\end{align*}
is concave in $(\delta,\gamma)$.
\end{lemma}
The proof of Lemma \ref{lemma_concave} is given in the Appendix.

Eventually, we obtain the concavity of the function $g(\delta,\gamma)$.
\begin{corollary}
\label{meme}
The function $g(\delta,\gamma)$ is concave in $(\delta,\gamma)$.
\end{corollary}
\begin{proof}
From Corollary \ref{cor_ohad} we have that $h^*(z)$ is a concave function.
Using Lemma \ref{lemma_concave} and the definition of $g(\delta,\gamma)$ we conclude that $g(\delta,\gamma)$  is concave in $(\delta,\gamma)$.
\end{proof}

The second main lemma shows that $\delta^*$ and $\gamma^*$ are optimal, which means that they maximize the function $g(\delta,\gamma)$.
First, we mention the KKT conditions adjusted to our problem:
\begin{lemma}[KKT conditions]
\label{cond}
Let $g(\delta,\gamma)$ be the objective function.
We consider the following optimization problem:
\begin{equation*}
\max_{\delta,\gamma} g(\delta,\gamma)
\end{equation*}
s.t.
\begin{equation*}
\gamma -1+z \leq 0,\quad -\gamma \leq 0,\quad \delta -z \leq 0,\quad - \delta \leq 0.
\end{equation*}
The Lagrangian of $g(\delta,\gamma)$ is $\mathcal{L}(\delta,\gamma,\lambda)= g(\delta,\gamma)-\lambda_1(\gamma-1+z)+\lambda_2\gamma-\lambda_3(\delta-z)+\lambda_4\delta$.
Since $g(\delta,\gamma)$ is a concave function then the following conditions are sufficient and necessary for optimality:
\begin{enumerate}[(1)]
\item $\frac{\partial g(\delta^*,\gamma^*)}{\partial \delta}=\lambda_3-\lambda_4$.
\item $\frac{\partial g(\delta^*,\gamma^*)}{\partial \gamma}=\lambda_1-\lambda_2$.
\item $\gamma\geq 0,\delta \geq 0$.
\item $\gamma-1+z\leq 0,\delta-z \leq 0$.
\item $\lambda_1,\lambda_2,\lambda_3,\lambda_4 \geq 0$.
\item $\lambda_1(\gamma-1+z)=\lambda_2\gamma=\lambda_3(\delta-z)=\lambda_4\delta=0$.
\end{enumerate}
\end{lemma}

The optimality conditions is a conclusion from the KKT conditions and Corollary \ref{meme}.
\begin{lemma}
\label{ce}
The following optimality conditions hold:
\begin{enumerate}[(a)]
\item
\label{f}
If $z\in [\frac{2a}{1+a},1]$, $\frac{\partial g(\delta^*,\gamma^*)}{\partial \delta}=0$, and
$\frac{\partial g(\delta^*,\gamma^*)}{\partial \gamma}>0$ then $\delta^*,\gamma^*$ are optimal.
\item
\label{s}
If $z\in [\frac{1-a}{1+a},\frac{2a}{1+a}]$, $\frac{\partial g(\delta^*,\gamma^*)}{\partial \delta}>0$, and
$\frac{\partial g(\delta^*,\gamma^*)}{\partial \gamma}>0$ then $\delta^*,\gamma^*$ are optimal.
\end{enumerate}
\end{lemma}
\begin{proof}
First, we consider case (\ref{f}) in which $z\in [\frac{2a}{1+a},1]$:\\
if we take $\lambda_2=\lambda_3=\lambda_4=0$, $\lambda_1=\frac{\partial g(\delta^*,\gamma^*)}{\partial \gamma}$, and $\gamma^*=1-z$, the KKT conditions, which are given in Lemma \ref{cond} hold since $\lambda_1\geq 0$.

Second, we consider case (\ref{s}) in which $z\in [\frac{1-a}{1+a},\frac{2a}{1+a}]$:\\
if we take  $\lambda_2=\lambda_4=0$, $\lambda_1=\frac{\partial g(\delta^*,\gamma^*)}{\partial \delta}$,
$\lambda_3=\frac{\partial g(\delta^*,\gamma^*)}{\partial \gamma}$, $\delta^*=z$, and $\gamma^*=1-z$, the KKT conditions hold since $\lambda_1\geq 0$ and $\lambda_3\geq 0$.
\end{proof}

Using Corollary \ref{meme} and Lemma \ref{ce} we prove Lemma \ref{mean1}.
\begin{proof}[Proof of Lemma \ref{mean1}]
From Corollary \ref{meme} we have that $g(\delta,\gamma)$ is concave in $(\delta,\gamma)$.
Thus, the KKT conditions are sufficient and necessary.
First we assume that $z\in [\frac{2a}{1+a},1]$.
We note that the expression $1+\frac{\delta-\gamma}{1+\delta-\gamma}$ is in $[\frac{2a}{1+a},1]$.
Furthermore, replacing $\gamma,\delta$ with $\gamma^*,\delta^*$ respectively, we find $\frac{1-z-\gamma}{1+\gamma-\delta}$ to be $0$.
We differentiate $g(\delta,\gamma)$ with respect to $\delta$ and evaluate it in $(\delta^*,\gamma^*)$:
\begin{eqnarray}
\frac{\partial g(\delta^*,\gamma^*)}{\partial \delta}&=&\frac{2\rho^*+\log_2(a)}{a-1} \nonumber 
\end{eqnarray}
Using basic algebra we find that the expression $\frac{2\rho^*+\log_2(a)}{a-1}$ is equal to zero iff $a^4-5a^3+6a^2-4a+1=0$.
Thus, setting $a\approx 0.4503$ to be the unique real root in the interval $[0,1]$ of the polynomial
$x^4-5x^3+6x^2-4x+1$ we establish that
\begin{eqnarray}
\frac{\partial g(\delta^*,\gamma^*)}{\partial \delta}=\frac{2\rho^*+\log_2(a)}{a-1}=0.
\end{eqnarray}
Now we differentiate $g(\delta,\gamma)$ with respect to $\gamma$ when $a\approx 0.4503$.
We find that the derivative is strictly positive
\begin{eqnarray}
\frac{\partial g(\delta^*,\gamma^*)}{\partial \gamma}&>&0.
\end{eqnarray}
Note that the derivative is positive when $a\leq 0.9$.
This can be seen since $\frac{\partial g(\delta^*,\gamma^*)}{\partial \gamma}$ is a monotonically increasing function of $a$, which for $a\leq 0.9$, equal to zero for $z < \frac{2a}{1+a}$.
Since we have found $a$ to be approximately $0.4503$ we have that $\frac{\partial g(\delta^*,\gamma^*)}{\partial \gamma}>0$.
Using Lemma \ref{ce} case (\ref{f}) we conclude that $\delta^*,\gamma^*$ are optimal.
%
The analysis for $z\in[0,\frac{1-a}{1+a}]$ is completely analogous.

Second, we assume $z\in[\frac{1-a}{1+a},\frac{2a}{1+a}]$.
In this case, we have that $h^*(z)=H_b(z)$. Using simple algebra we obtain that
\begin{eqnarray}
\frac{\partial g(\delta^*,\gamma^*)}{\partial \delta}&>&0 \nonumber \\
\frac{\partial g(\delta^*,\gamma^*)}{\partial \gamma}&>&0.
\end{eqnarray}
Using Lemma \ref{ce} case (\ref{s}) we conclude that $\delta^*,\gamma^*$ are optimal.
Thus, for all $z\in[0,1]$ we have that $\delta^*,\gamma^*$ are optimal.
\end{proof}

We now prove Theorem \ref{main_result}.
In the proof we show that $Th^*(z)=h^*(z)+\rho^*$, $\rho^*=\frac{2H_b(z)}{3+a}$, $\delta,\gamma$ which maximize the operator $T$ are $\delta^*,\gamma^*$, and $h^*(z)$ is given in (\ref{thefunction}).
This implies that $h^*(z)$ solves the Bellman Equation.
Therefore, $\rho^*$ is the optimal average reward, which is equal to the channel capacity.

%

First, we prove Theorem \ref{main_result}(a).
\begin{proof}[Proof of Theorem \ref{main_result}(a)]
We would like to prove that the capacity of the the Ising channel with feedback is
$C_f=\left(\frac{2H(a)}{3+a} \right)\approx 0.5755$ where $a\approx 0.4503$ is a specific root of
the fourth-degree polynomial $x^4-5x^3+6x^2-4x+1$.
According to Theorem \ref{Bellman}, if we identify a scalar $\rho$ and a bounded function $h(z)$ such that
\begin{equation}
\rho + h(z) = \sup_{u \in {\cal U}}
\left(g(z,u) + \int P_w(dw | z,u) h(F(z,u,w))\right) \: \: \forall
z\in{\cal Z}
\end{equation}
then $\rho=\rho^*$.
Using Lemma \ref{mean1} we obtain that $\delta^*$ and $\gamma^*$, as defined in (\ref{gamma}),(\ref{delta}), maximize $Th^*(z)$ when $a\approx 0.4503$.
In addition, we show that $h^*(z)$, which is defined in (\ref{thefunction}), satisfies the Bellman Equation, $Th^*(z)=h^*(z)+\rho^*$, where $\rho^*=\frac{2H(a)}{3+a}$.
This follows from Lemma \ref{mid_section_h} and Lemma \ref{end-section-h}.
Therefore, we have identified a bounded function, $h^*(z)$, and a constant, $\rho^*$, together with a policy,
$\gamma^*,\delta^*$, which satisfy the Bellman Equation.
Thus, the capacity of the Ising channel with feedback is
$\rho^*=h^*(0)=h^*(1)=\frac{2H(a)}{3+a}\approx 0.575522$.
\end{proof}

Now we prove Theorem \ref{main_result}(b). The proof is straightforward.
\begin{proof}[Proof of Theorem \ref{main_result}(b)]
We define $g(z)=\frac{2H_b(z)}{3+z}$ and we calculate $g'(z)=\frac{8\log_2(1-z)-6\log_2(z)}{(3+z)^2}$.
$8\log_2(1-z)-6\log_2(z)=0$ iff $(1-z)^8-z^6=0$.
The polynomial $(1-z)^8-z^6=0$ is reducible, hence we can write $(1-z)^8-z^6=(1-4z+6z^2-3z^3+z^4)(1-4z+6z^2-5z^3+z^4)$.
Therefore, $g'(a)=0$ since $a\approx 0.4503$ is the root of the polynomial $x^4-5x^3+6x^2-4x+1$.
It is easy to verify that $g'(a-\epsilon)>0$ and $g'(a+\epsilon)<0$.
Together with the fact that $a$ is the only real number in $[0,1]$ which sets $g'(z)$ to zero,
$a$ is a maximum point of $g(z)$, for $0\leq z \leq 1$.
\end{proof}
\section{Relation of the DP Results and the Coding Scheme}
\label{seceight}
In this section we analyse the DP results and derive the coding scheme from these results.
Especially, we use the histogram of $z$, which is presented in Fig. \ref{4plots}.
However, we first recap a few definitions that where used in this paper:
\begin{enumerate}
\item $z_t=p(s_t=0|y^t)$, where $s_t=x_t$.
This means that $z_t$ is the probability of the input $x_t$ being $0$ given the output.
Thus, if $z_t=0$, then $x_t=1$ with probability 1.
\item  We defined $\delta=z_{t-1}u_t(1,1)$ and $\gamma=(1-z_{t-1})u_t(2,2)$, where $u_t(1,1)=\Pr(x_t=0|s_{t-1}=0)$ and $u_t(2,2)=\Pr(x_t=1|s_{t-1}=1)$.
\item We established that
\begin{equation}
\gamma^*(z)=\left\{ \begin{array}{ll}
a+az,& \text{if } z\in [p_0,p_1]\\
1-z,& \text{if } z\in [p_1,p_3]
\end{array} \right.
\end{equation}
\\
\begin{equation}
\delta^*(z)=\left\{ \begin{array}{ll}
z,& \text{if } z\in [p_0,p_2]\\
a(2-z),& \text{if } z\in [p_2,p_3].
\end{array} \right.
\end{equation}
\item We also established the equation in (\ref{11}):
\begin{equation}
z_t=\left\{ \begin{array}{ll}
1+\frac{\delta_t -z_{t-1}}{1+\delta_t -\gamma_t},& \text{if } w_t=0\\
\frac{1-z_{t-1}-\gamma_t}{1+\gamma_t -\delta_t},& \text{if } w_t=1.
\end{array} \right.
\end{equation}

\end{enumerate}

We also remind that in the histogram, which is presented in Fig. \ref{4plots}, $z_t$ alternates between four points, two of which are $0$ and $1$.
In order to keep in mind that these points stand for probability we denote them as $p_0,p_1,p_2,p_3$, where $p_0=0$ and $p_3=1$.
Using (\ref{11}) and the definition of $\gamma^*,\delta^*$ we can derive Table \ref{code_table}.
The table presents $z_{t+1}$ as a function of $z_t$ and $y_{t+1}$.
It also presents the optimal action parameters, $u_t(1,1),u_t(2,2)$, for each state.
The action parameters are calculated from the parameters $\delta^*,\gamma^*$.

\begin{table}[h!]
\caption{The DP states at time $t+1$ as a function of the previous state and the output calculated using (\ref{11}).
         The table presents the optimal actions for each state.} \centering
\label{code_table}
\begin{tabular}[h]{||c||c|c|c|c||}
\hline \hline
    & $z_t=p_0$ & $z_t=p_1$ & $z_t=p_2$ & $z_t=p_3$\\
\hline \hline
$y_t=0$ & $z_{t+1}=p_3$ & $z_{t+1}=p_3$ & $z_{t+1}=p_3$ & $z_{t+1}=p_2$\\
\hline
$y_t=1$ & $z_{t+1}=p_1$ & $z_{t+1}=p_0$ & $z_{t+1}=p_0$ & $z_{t+1}=p_0$\\
\hline
$u_{t+1}(2,2)$ & $a$ & $1$ & $1$ & irrelevant\\
\hline
$u_{t+1}(1,1)$ & irrelevant & $1$ & $1$ & $a$\\
\hline \hline
\end{tabular}.
\end{table}

Assume at first that at time $t-1$ the state is $z_{t-1}=p_0=0$.
\begin{enumerate}
\item
\emph{Decoder:} Using the definition of $z_{t-1}$ we deduce that $p(s_{t-1}=0|y^{t-1})=0$ and hence $x_{t-1}=1$ with probability $1$.
Thus, the decoder decodes $1$.
\item
\emph{Encoder:} The optimal actions are $\delta_t^*(0)=0$ and $\gamma_t^*(0)=a$.
Using the definition of $\gamma^*$ we conclude that $\Pr(x_t=1|s_{t-1}=1)=a$.
Thus, $\Pr(x_t=0|s_{t-1}=1)=1-a$, which means that, given that $s_{t-1}=x_{t-1}=1$, the probability to send $1$ again is $a$.
This result gives us the alternation probability from $1$ to $0$, which is $1-a$.
Since $s_{t-1}=x_{t-1}=1$ with probability $1$, the action parameter $\delta^*$ is irrelevant because it concerns the case in which $s_{t-1}=0$.
Indeed, using the definition of $\delta^*$, we can see that $\delta_t^*(0)=0$.
\end{enumerate}

We now use Table \ref{code_table} in order to find the next state.
We have two options; if the output is $0$ we move to state $p_3=1$.
For this state the analysis is similar to the state $p_0$, switching between $0$ and $1$.
Note that since the next state is $p_3=1$ the decoder decodes the bit which was sent.
If, on the other hand, the output is $1$ we move to the state $z_t=p_1$.
Assuming $z_t=p_1$ we have the following:
\begin{enumerate}
\item
\emph{Decoder:} Using the definition of $z_t$ we deduce that $p(s_{t}=0|y^{t})=p_1$ and hence $x_{t}=1$ with probability $p_1$.
Thus, the decoder does not decode and waits for the next bit.
\item
\emph{Encoder:} The optimal actions are $\delta_{t+1}^*(p_1)=p_1$ and $\gamma_{t+1}^*(p_1)=1-p_1$.
Using the definition of $\gamma^*$ we conclude that $\Pr(x_{t+1}=1|s_{t}=1)=1$
and using the definition of $\delta^*$ we conclude that $\Pr(x_{t+1}=0|s_{t}=0)=1$.
This means that $x_{t+1}=s_t=x_t$ with probability $1$.
\end{enumerate}
The analysis for state $p_2$ is done in a similar way.
\begin{figure}[h!]

\begin{center}
\begin{psfrags}
    \psfragscanon
        \psfrag{A}[][][0.9]{DP state $p_0$}
        \psfrag{B}[][][0.9]{ }
    \psfrag{C}[][][0.9]{De: $p(x_t=0|y^t)=0$}
    \psfrag{D}[][][0.9]{En: $\Pr(x_{t+1}=x_t)=a$}
        \psfrag{E}[][][0.9]{DP state $p_3$}
    \psfrag{F}[][][0.9]{ }
    \psfrag{G}[][][0.9]{De: $p(x_t=0|y^t)=1$}
    \psfrag{H}[][][0.9]{En: $\Pr(x_{t+1}=x_t)=a$}
        \psfrag{I}[][][0.9]{DP states $p_1,p_2$}
    \psfrag{J}[][][0.9]{ }
      \psfrag{K}[][][0.9]{De: Waits}
        \psfrag{L}[][][0.9]{En: $x_{t+1}=x_t$}
    \psfrag{m}[][][0.9]{$y_{t+1}=1$}
      \psfrag{n}[][][0.9]{$y_{t+1}=0$}
      \psfrag{u}[][][0.9]{$y_{t+1}=1$}
    \psfrag{p}[][][0.9]{$y_{t+1}=1$}
    \psfrag{q}[][][0.9]{$y_{t+1}=0$}
    \psfrag{r}[][][0.9]{$y_{t+1}=0$}

\includegraphics[width=14cm]{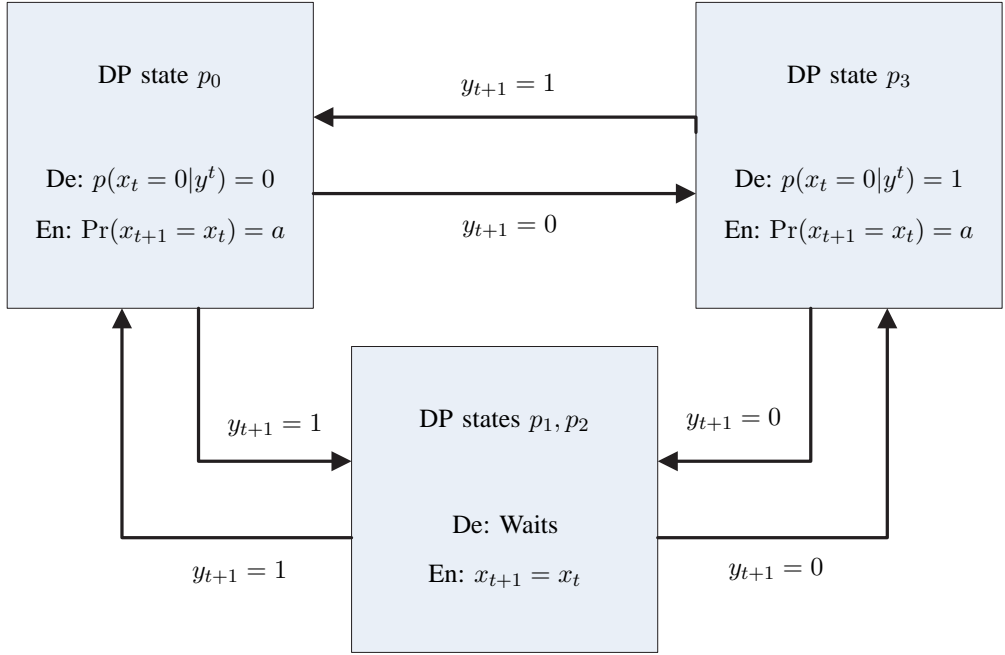}
\caption{The coding graph for the capacity-achieving coding scheme at time $t$. The states,
$p_i\quad i=0,1,2,3$ are the DP states where $p_0=0,p_3=1$. The labels on the arcs represents the
output of the channel at time $t+1$. The decoder and the encoder rules, which are written in
vertices of the graph, yield the coding scheme presented in Theorem \ref{main_result_code}. }
\label{coding_graph} \psfragscanoff
\end{psfrags}
\end{center}
\end{figure}


We can now create a coding graph for the capacity-achieving coding scheme.
Decoding only when the states are $p_0$ or $p_3$ results in a zero-error decoding.
The coding graph is presented in Fig. \ref{coding_graph}.
In the figure we have three vertices, which corresponds to the DP states.
At each vertex we mention the corresponding state or states, the decoder action, and the encoder action.
The edges' lables are the output of the channel.
The edges from vertices $p_0$ to $p_3$ and vice versa corresponds to case $(1.1)$ in Theorem \ref{main_result_code} in the encoder scheme and
to case $(1.1)$ in Theorem \ref{main_result_code} in the decoder scheme.
The edges between vertices $p_0$ and $p_1,p_2$ and between $p_3$ and $p_1,p_2$ correspond to case $(1.2)$ in Theorem \ref{main_result_code}
in the encoder scheme and to case $(1.2)$ in Theorem \ref{main_result_code} in the decoder scheme.

\section{Capacity-Achieving Coding Scheme Analysis}
\label{secseven}

In this section we show that the coding scheme presented in Theorem \ref{main_result_code} and in the previous section indeed achieves the capacity.
In order to analyze the coding scheme, we would like to present the same coding scheme in a different way.
One can see that the scheme that presented in Theorem \ref{main_result_code} resembles to the zero-error coding scheme
for the Ising channel without feedback, which has been found by Berger and Bonomi
\cite{CapacityandzeroerrorcapacityofIsingchannel}.
The zero-error capacity-achieving coding scheme for the Ising channel without feedback is simple:
the encoder sends each bit twice and the decoder considers every other bit.
The channel topology ensures that every second bit is correct with probability $1$.
Moreover, it is clear that this scheme achieves the zero-error capacity of the Ising channel
without feedback, which is $0.5$ bit per channel use.
Thus, one can see the intuition behind the coding scheme suggested in Theorem \ref{main_result_code}.

The following proof shows that the coding scheme presented in Theorem \ref{main_result_code} indeed achieves
the capacity. In the proof we calculate the expected length of strings in the channel input and divide
it by the expected length of strings in the channel output.

\begin{proof}[Proof of Theorem \ref{main_result_code}]
Let us consider an encoder that contains two blocks, as in Fig. \ref{encoder}.
The first block is a data encoder. The data encoder receives a message $M^n$ ($M=\{0,1\}$) of length $n$
distributed i.i.d. Bernoulli $\left(\frac{1}{2}\right)$ and transfers it to a string of data
with probability of alternation between $1$ and $0$ and vice versa of $q$.
This means that if some bit is $0$ (alternatively $1$), the next bit is $1$ (alternatively $0$) with probability $q$.
In order to create a one-to-one correspondence between the messages and the data strings we need
the data strings to be longer than $n$.

\begin{figure*}[h!]
\begin{center}
\begin{psfrags}
    \psfragscanon
    \psfrag{a}[][][0.9]{Encoder}
    \psfrag{b}[][][0.9]{Data}
        \psfrag{h}[][][0.9]{Encoder}
    \psfrag{c}[][][0.9]{Channel}
        \psfrag{i}[][][0.9]{Encoder}
    \psfrag{d}[][][0.9]{$m$}
    \psfrag{e}[][][0.9]{$d_t$}
    \psfrag{f}[][][0.9]{$x_t$}
    \psfrag{g}[][][0.9]{$y_{t-1}$}

\includegraphics[width=15cm]{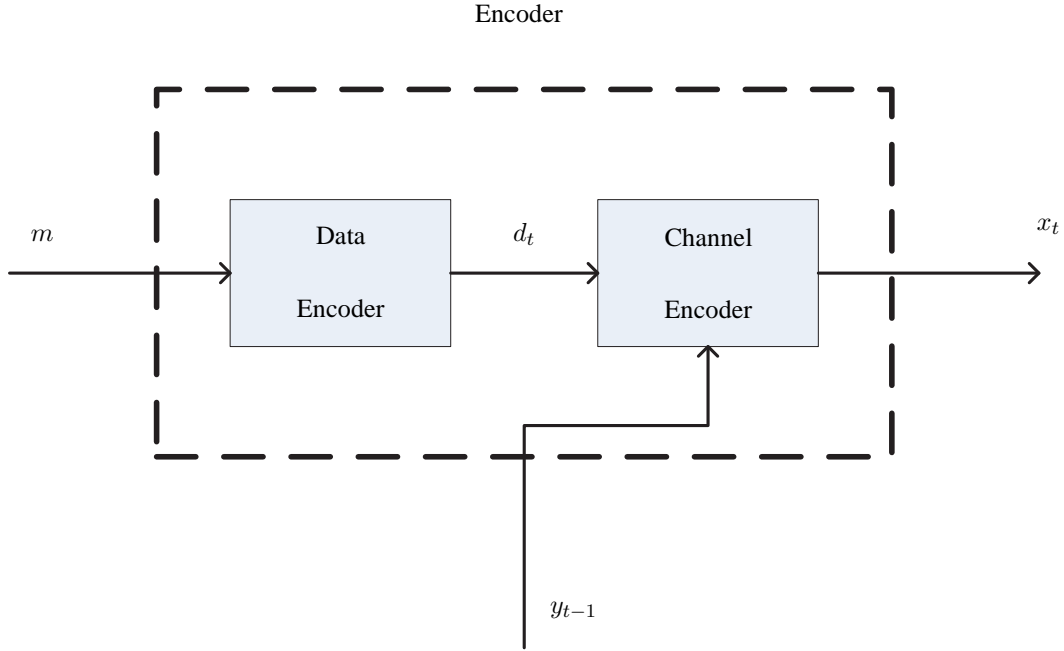}
\caption{The channel encoder block which consists of two sub-encoders. One block encodes data and the other
performs the channel encoding.  }
\psfragscanoff
\label{encoder}
\end{psfrags}
\end{center}
\line(1,0){515}
\end{figure*}
Let us calculate the length of the data strings needed in order to transform the message into
a string with alternation probability of $q$.
We notice that $p(x_t|x^{t-1})=p(x_t|x_{t-1})$. Thus, the entropy rate is
\begin{eqnarray}
\lim_{n\rightarrow \infty} \frac{H\left(X^n\right)}{n} &\stackrel{(a)}{=}& \lim_{n\rightarrow \infty}
\frac{1}{n}\sum_{i=1}^{n}H\left(X_i|X_{i-1}\right) \nonumber \\
&=& H\left( X_i|X_{i-1}\right) \stackrel{(b)}{=} H\left(q\right)
\end{eqnarray}
where
\begin{itemize}
\item $(a)$ is due to the chain rule and since $p(x_t|x^{t-1})=p(x_t|x_{t-1})$.
\item $(b)$ is due to the fact that the probability of alternation is $q$.
\end{itemize}
Therefore, given a message of length $n$, the data encoder transfers it into a data string of length
$\frac{n}{H_b(q)}$. This can be done using the method of types for the binary sequence.
Given a probability $q$ and a binary sequences of length $n$, the size of the typical set is about $2^{nH_b(q)}$.
Hence, we can map the set of Bernoulli $\frac{1}{2}$ sequences of length $n$ (which is of size $2^n$) to the set
of sequences of length $\frac{n}{H_b(q)}$ with alternation probability $q$ (which is of size $2^{\frac{n}{H_b(q)}H_b(q)}=2^n$).
One can also use the mapping presented in \cite{enumerative}.
This mapping gives a simple way to enumerate the indexes of Markov sequences of length $\frac{n}{H_b(q)}$ and
of the binary sequences of length $n$, which distributed Bernoulli $(\frac{1}{2})$.
One can enumerate these sequences and establish a mapping from the Bernoulli sequences to the Markov sequences simply by matching their indexes.

The second block is the channel encoder. This encoder receives a data string in which the probability
of alternation between $0$ to $1$ and vice versa is $q$.
This sequence passes through the encoder, which sends some bits once and some bits twice. Due to that property, the transmitted
bit at time $t$ is not necessarily the data bit at the $t$th location. This is why the encoder scheme uses two time indexes, $t$
and $t'$, which denote the data bit location and the current transmission time, respectively.
The encoder works as mentioned in Theorem \ref{main_result_code}:
\begin{enumerate}[(1)]
\item
\emph{Encoder:} At time $t'$, the encoder knows $s_{t'-1}=x_{t'-1}$; it is clear from the encoder description that $x_{t'-1}=m_{t-1}$ and we send the bit $m_t$ $(x_{t'}=m_t)$:
\begin{enumerate}[({1.}1)]
\item If $y_{t'}\neq s_{t'-1}$ then move to the next bit, $m_{t+1}$. This means that we send $m_t$ once.
The probability that $y_{t'}\neq s_{t'-1}$ is the probability that $x_{t'}\neq x_{t'-1}$ and $y_{t'}=x_{t'}$, namely $q\cdot \frac{1}{2}$.
\item If $y_{t'}=s_{t'-1}$ then $x_{t'}=x_{t'+1}=m_t$, which means that the encoder sends $m_t$ twice (at time $t'$ and $t'+1$) and then move to the next bit.
The probability that $y_{t'}=s_{t'-1}$ is the probability that $x_{t'}= x_{t'-1}$ or $x_{t'}\neq x_{t'-1}$ and $y_{t'}\neq x_{t'}$, namely $(1-q)+\frac{1}{2}q=\frac{2-q}{2}$.
\end{enumerate}
\end{enumerate}

Now we calculate the expected length of the channel encoder output string.
First, the message is of length $n$ and distributed Bernoulli $\frac{1}{2}$.
Thus, the length of the string which has alternation probability of $q$ is $\frac{n}{H(q)}$.
Then, with probability $\frac{2-q}{2}$ we send two bits and with probability $\frac{q}{2}$ we send one bit.
Hence, the expected length of the channel encoder output string is
$\frac{n}{H_b(q)}\left(2\frac{2-q}{2}+\frac{q}{2}\right)=\frac{4-q}{2}\frac{n}{H_b(q)}$.

We send $2^n$ messages in $\frac{1-q}{2}\frac{n}{H_b(q)}$ transmissions, hence the rate is
$\frac{n}{\frac{4-q}{2}\frac{n}{H_b(q)}}=\frac{2H(q)}{4-q}$.
Setting $q=1-a$ we achieve the rate $\frac{2H_b(1-a)}{3+a}=\frac{2H(a)}{3+a}$.
This is true for any $a\in [0,1]$, in particular it holds for the
unique positive root in $[0,1]$ of the polynomial $x^4-5x^3+6x^2-4x+1$.
Using Theorem \ref{main_result}, the expression $\frac{2H(a)}{3+a}$ is equal to the capacity
of the Ising channel with feedback. This means that the scheme achieves the capacity.
\end{proof}

The following illustration shows the channel encoding and decoding scheme.
Assume we are to send the data string $0110$ and the state at time $t'=0$ is $0$.
Following the channel encoder rules, we present the encoder output, $x_{t'}$, in Table \ref{encoding}.
The table shows in each time, $t'=1,2,3,\cdots $, the channel input $x_{t'}$ and output $y_{t'}$.
The right-most column of the table refers to the data string bit that is currently encoded (denoted with the time
index $t$ in the decoder scheme), from first to fourth.
\begin{table}[h]
\caption{Example for encoding the word $0110$ where we assume channel output using the channel topology.
The "encoded bit" represents the bit we currently encode ($1,2,3,4$).} \centering
\label{encoding}
\begin{tabular}[h]{||c|c|c|c|c|c||}
\hline \hline
time & channel & channel & channel & encoded & case \\
 $t'$& state & input & output & bit & \\
 & $s_{t'}=x_{t'-1}$& $x_{t'}$& $y_{t'}$& $t$ & \\
\hline \hline
1 & 0 & 0 & 0 (w.p 1) & 1 & 1.2\\
\hline
2 & 0 & 0 & 0 (w.p 1) & 1 & 1.2\\
\hline
3 & 0 & 1 & 1 (w.p $\frac{1}{2}$) & 2 & 1.1\\
\hline
4 & 1 & 1 & 1 (w.p 1) & 3 & 1.2\\
\hline
5 & 1 & 1 & 1 (w.p 1) & 3 & 1.2\\
\hline
6 & 1 & 0 & 1 (w.p $\frac{1}{2}$)) & 4 & 1.2\\
\hline
7 & 0 & 0 & 0 (w.p 1) & 4 & 1.2\\
\hline \hline
\end{tabular}
\end{table}
In this example the input to the channel is $0011100$ and the output of the channel is $0011110$.
Note that in the decoding process on the 3rd step, the channel state is $0$ and the channel input is $1$;
since the channel output was also $1$ this bit was sent only once.
Now we follow the decoder rules in order to decode the received word $0011110$.
We remind that the decoder rules are as follows:
\begin{enumerate}[(1)]
\item
\emph{Decoder:} At time $t'$, assume the state $s_{t'-1}$ is known at the decoder and we are to decode the bit $\hat{m}_t$:
\begin{enumerate}[({1.}1)]
\item If $y_{t'}\neq s_{t'-1}$ then $\hat{m}_t=y_{t'}$ and $s_{t'}=y_{t'}$.
\item If $y_{t'}=s_{t'-1}$ then wait for $y_{t'+1}$, $\hat{m}_t=y_{t'+1}$ and $s_{t'}=y_{t'+1}$.
\end{enumerate}
\end{enumerate}
Table \ref{decoding} presents the decoder decisions made in each time $t$.
The second column represents the channel output (the decoder input), $y_t$.
In the third column, the channel state, $s_t$, from the decoder point of view is presented.
The question mark stands for an unknown channel state. In this situation, the decoder cannot decode the
bit that was sent and it has to wait for the next bit.
The action column records, in each time, the action made by the decoder, which can decode or take no action
and wait for the next bit to arrive. In the last column the decoded string is presented.
\begin{table}[h!]
\caption{Example for decoding the word $0011110$ where we use the decoding rules.
The channel state is given from the decoder's point of view.
Since the decoder decodes only when the state is known with probability $1$, we denote bits we cannot yet decode by a question mark.} \centering
\label{decoding}
\begin{tabular}[h]{||c|c|c|c|c|c||}
\hline \hline
time & channel & channel & action & decoded & case\\
 $t'$& output $y_{t'}$ & state $s_{t'}$ &  & word & \\
\hline \hline
1 & 0 & ? & none & ? & 1.2\\
\hline
2 & 0 & 0 & decode & 0 & 1.2\\
\hline
3 & 1 & 1 & decode & 01 & 1.1\\
\hline
4 & 1 & ? & none & 01? & 1.2\\
\hline
5 & 1 & 1 & decode & 011 & 1.2\\
\hline
6 & 1 & ? & none & 011? & 1.2\\
\hline
7 & 0 & 0 & decode & 0110 & 1.2\\
\hline \hline
\end{tabular}
\end{table}
Following the decoder rules, we have decoded the correct word $0110$.
As we can see, using this coding scheme we can decode the word instantaneously with no errors.

An interesting fact is that in order to achieve the capacity using this coding scheme, we do not need
to use the feedback continuously.
It is enough to use the feedback only when there is an alternation between $0$ to $1$ (or vice versa) in the bits we send.
When there is no alternation, the feedback is not needed since the bit is sent twice regardless of the channel output.
Several cases of partial feedback use are studied in \cite{tofeedornottofeed}.

\section{Conclusions}
\label{secnine}

We have derived the capacity of the Ising channel, analyzed it and presented a simple capacity-achieving
coding scheme.
As an immediate result of this work we can tighten the upper bound for the capacity of the one-dimensional
Ising Channel to be $0.575522$, since the capacity of a channel without feedback cannot exceed the capacity
of the same channel with feedback.

A DP method is used in order to find the capacity of the Ising channel with feedback.
In the case presented in this paper, we have also established a connection between the DP results and the capacity-achieving coding scheme.
An interesting question that arises is whether there exists a general method for finding the capacity for
two states channels with feedback, whose states are a function of the previous state, the input,
and the previous output. It may be the case that the solution of the DP for such a channel has a fixed pattern.
Towords this goal, a new coding scheme is provided in \cite{new} for unifilar finite state channels that is based on posterior matching.

\appendix

\begin{proof}[Proof of Lemma \ref{lemma_ohad}]
the function $\eta(x)$ is continuous by definition, since
$f$ and $g$ are continuous and $f(\beta)=g(\beta)$.
We continue the function $f(x)$ on $[\beta,\gamma]$ with a straight line
with incline $f'_-(\beta)$ and the function $g(x)$ with a straight line
with incline $g'_+(\beta)$ as in Fig. \ref{concatenate}. We define
\begin{equation}
f_1(x)=\left\{ \begin{array}{ll}
f(x),& \text{if } x\in [\alpha,\beta]\\
(x-\beta)f'_-(\beta)+f(\beta),& \text{if } x\in [\beta,\gamma]
\end{array} \right.
\end{equation}
and
\begin{equation}
g_1(x)=\left\{ \begin{array}{ll}
(x-\beta)g'_+(\beta)+g(\beta),& \text{if } x\in [\alpha,\beta]\\
g(x),& \text{if } x\in [\beta,\gamma].
\end{array} \right.
\end{equation}
The functions $f_1(x)$ and $g_1(x)$ are concave since we continue the functions with a straight line.
Since $f(x)$ and $g(x)$ are concave we have $f(x)\leq (x-\beta)f'_-(\beta)+f(\beta)$ for all $x\in [\alpha,\beta]$
and $g(x)\leq (x-\beta)g'_+(\beta)+g(\beta)$ for all $x\in [\beta,\gamma]$.
Hence, for all $x\in [\alpha,\beta]$ we have $g_1(x)\geq f(x)$ and for all $x\in [\beta,\gamma]$ we have
$f_1(x)\geq g(x)$.
Therefore, $\eta(x)=\min \{f_1(x),g_1(x)\}$ and, since the minimum of two concave functions
is a concave function\cite{convex_optimization_boyd}, $\eta(x)$ is concave.
\end{proof}

\begin{figure*}[t]
\begin{center}
\begin{psfrags}
    \psfragscanon
    \psfrag{a}[][][0.9]{Function 1}
    \psfrag{b}[][][0.9]{Function 2}
    \psfrag{c}[][][0.9]{$f(x)$}
        \psfrag{d}[][][0.9]{$g(x)$}
    \psfrag{e}[][][0.9]{$f_1(x)$}
    \psfrag{f}[][][0.9]{$g_1(x)$}
        \psfrag{g}[][][0.9]{concatenation point}
        \psfrag{h}[][][0.9]{$\alpha$}
    \psfrag{i}[][][0.9]{$\beta$}
        \psfrag{j}[][][0.9]{$\gamma$}

\includegraphics[width=10cm]{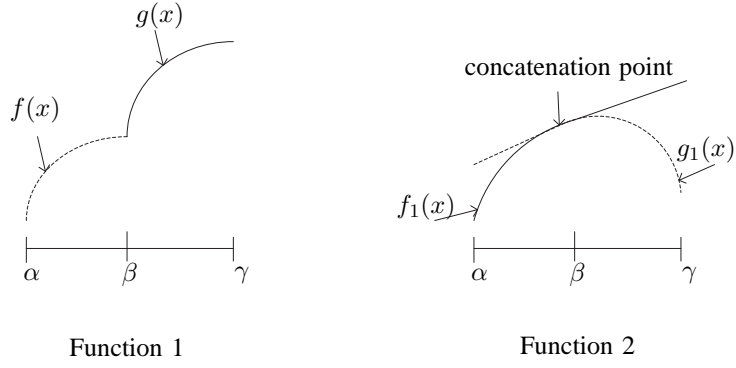}
\caption{Examples for concatenation of two continuous, concave functions. It is easy to see from the figures
the intuition behind Lemma (\ref{lemma_ohad}) . Function 1 is not concave due to the fact that $f'_-(\beta)<g'_+(\beta)$.
Function 2 is concave since $f'_-(\beta)=g'_+(\beta)$. }
\psfragscanoff
\label{concatenate}
\end{psfrags}
\end{center}
\line(1,0){515}
\end{figure*}

\begin{proof}[Proof of Lemma \ref{lemma_concave}]
We now show that the expression
\begin{align*}
H\left(\frac{1}{2}+\frac{\delta-\gamma}{2}\right) +\delta +\gamma -1
+\frac{1+\delta-\gamma}{2}h\left( 1+\frac{\delta-z}{\delta+1-\gamma}\right)
+\frac{1-\delta+\gamma}{2} h\left(\frac{1-z-\gamma}{1+\gamma-\delta}\right)
\end{align*}
is concave in $(\delta,\gamma)$.

To show this we use the fact that $h(z)$ is a concave function in $z$.
Note that, since the binary entropy is concave, the expression $H\left(\frac{1}{2}+\frac{\delta-\gamma}{2}\right) +\delta +\gamma -1$
is concave in $(\delta,\gamma)$.
We examine the expression $\frac{1+\delta-\gamma}{2}h\left( 1+\frac{\delta-z}{\delta+1-\gamma}\right)$.
Let us denote $\eta_i=\frac{1+\delta_i-\gamma_i}{2}\ ,\quad i=1,2$.
For every $\alpha\in [0,1]$ we obtain that
\begin{equation*}
\frac{\alpha\eta_1}{\alpha\eta_1+(1-\alpha)\eta_2}h\left( 1+\frac{\delta_1-z}{\eta_1}\right)+
\frac{(1-\alpha)\eta_2}{\alpha\eta_1+(1-\alpha)\eta_2}h\left( 1+\frac{\delta_2-z}{\eta_2}\right)\stackrel{(i)}{\leq}
h\left( 1+\frac{\alpha\delta_1+(1-\alpha)\delta_2-z}{\alpha\eta_1+(1-\alpha)\eta_2}\right),
\end{equation*}
where $(i)$ is due to the fact that $h(z)$ is concave.
This result implies that
\begin{equation*}
\alpha\eta_1 h\left( 1+\frac{\delta_1-z}{\eta_1}\right)+(1-\alpha)\eta_2 h\left( 1+\frac{\delta_2-z}{\eta_2}\right)\leq
\left(\alpha\eta_1 +(1-\alpha)\eta_2\right)h\left( 1+\frac{\alpha\delta_1+(1-\alpha)\delta_2-z}{\alpha\eta_1+(1-\alpha)\eta_2}\right).
\end{equation*}
Hence, $\frac{1+\delta-\gamma}{2}h\left( 1+\frac{\delta-z}{\delta+1-\gamma}\right)$ is concave in $(\delta,\gamma)$.
It is completely analogous to show that the expression $\frac{1-\delta+\gamma}{2}h\left( 1+\frac{1-z-\gamma}{1-\delta+\gamma}\right)$ is also concave in $(\delta,\gamma)$.
Thus, we derive that the expression
\begin{align*}
H\left(\frac{1}{2}+\frac{\delta-\gamma}{2}\right) +\delta +\gamma -1
+\frac{1+\delta-\gamma}{2}h\left( 1+\frac{\delta-z}{\delta+1-\gamma}\right)
+\frac{1-\delta+\gamma}{2} h\left(\frac{1-z-\gamma}{1+\gamma-\delta}\right)
\end{align*}
is concave in $(\delta,\gamma)$.
\end{proof}
\bibliographystyle{IEEEtran}
\bibliography{IEEEabrv,ref}

\begin{thebibliography}{10}
\providecommand{\url}[1]{#1}
\csname url@samestyle\endcsname
\providecommand{\newblock}{\relax}
\providecommand{\bibinfo}[2]{#2}
\providecommand{\BIBentrySTDinterwordspacing}{\spaceskip=0pt\relax}
\providecommand{\BIBentryALTinterwordstretchfactor}{4}
\providecommand{\BIBentryALTinterwordspacing}{\spaceskip=\fontdimen2\font plus
\BIBentryALTinterwordstretchfactor\fontdimen3\font minus
  \fontdimen4\font\relax}
\providecommand{\BIBforeignlanguage}[2]{{%
\expandafter\ifx\csname l@#1\endcsname\relax
\typeout{** WARNING: IEEEtran.bst: No hyphenation pattern has been}%
\typeout{** loaded for the language `#1'. Using the pattern for}%
\typeout{** the default language instead.}%
\else
\language=\csname l@#1\endcsname
\fi
#2}}
\providecommand{\BIBdecl}{\relax}
\BIBdecl

\bibitem{isingmodel1}
W.~Lenz, ``Beitrage zum verstandnis der magnetischen eigenschaften in festen
  korpern,'' \emph{Physikalische Zeitschrift}, vol.~21, 1920.

\bibitem{isingmodel2}
E.~Ising, ``Beitrag zur theorie des ferromagnetismus,'' \emph{Zeitschrift for
  Physik A Hadrons and Nuclei}, vol.~31, 1925.

\bibitem{twodimisingchannel}
L.~Onsager, ``Crystal statistics. i. a two-dimensional model with an
  order-disorder transition,'' \emph{Phys. Rev.}, vol.~65, Feb 1944.

\bibitem{CapacityandzeroerrorcapacityofIsingchannel}
T.~Berger and F.~Bonomi, ``Capacity and zero-error capacity of ising
  channels,'' \emph{IEEE Transactions on Information Theory}, vol.~36, no.~7,
  pp. 173--180, 1990.

\bibitem{arimoto}
S.~Arimoto, ``An algorithm for computing the capacity of arbitrary discrete
  memoryless channels,'' \emph{IEEE Transactions on Information Theory}, 1972.

\bibitem{blahut}
R.~Blahut, ``Computation of channel capacity and rate-distortion functions,''
  \emph{IEEE Transactions on Information Theory}, 1972.

\bibitem{Shannon49waterfilling}
C.~E. Shannon, ``Communication in the presence of noise,'' \emph{Proceedings of
  the IRE}, vol.~37, no.~1, pp. 10--21, 1949.

\bibitem{Pinsker60}
M.~S. Pinsker, \emph{Information and Information Stability of Random Variables
  and Processes}.\hskip 1em plus 0.5em minus 0.4em\relax Moskva: Izv. Akad.
  Nauk, 1960, translated by A. Feinstein, 1964.

\bibitem{Kim06MA}
Y.~H. Kim, ``Feedback capacity of the first-order moving average {G}aussian
  channel,'' \emph{IEEE Transactions on Information Theory}, vol.~52, pp.
  3063--3079, 2006.

\bibitem{goldsmith96capacity}
A.~Goldsmith and P.~Varaiya, ``Capacity, mutual information, and coding for
  finite-state {M}arkov channels,'' \emph{IEEE Transactions on Information
  Theory}, vol.~4, pp. 868--886, 1996.

\bibitem{chenberger}
J.~Chen and T.~Berger, ``The capacity of finite-state markov channels with
  feedback,'' \emph{IEEE Transactions on Information Theory}, march 2005.

\bibitem{Permuter06_trapdoor_submit}
H.~H. Permuter, P.~W. Cuff, B.~Van-Roy, and T.~Weissman, ``Capacity of the
  trapdoor channel with feedback,'' \emph{IEEE Transactions on Information
  Theory}, vol.~54, pp. 3150--3165, 2008.

\bibitem{causalityfeedbackanddirectedinformation}
J.~Messy, ``Causality, feedback and directed information,'' in \emph{IEEE
  International Symposium on Information Theory and Applications}, 1990.

\bibitem{Marko73}
H.~Marko, ``The bidirectional communication theory- a generalization of
  information theory,'' \emph{IEEE Transactions on Information Theory}, vol.
  COM-21, pp. 1335--1351, 1973.

\bibitem{Kramer03}
G.~Kramer, ``Capacity results for the discrete memoryless network,'' \emph{IEEE
  Transactions on Information Theory}, vol. IT-49, pp. 4--21, 2003.

\bibitem{Kim07feedback}
Y.~H. Kim, ``A coding theorem for a class of stationary channels with
  feedback,'' \emph{IEEE Transactions on Information Theory}, vol.~25, pp.
  1488--1499, April, 2008.

\bibitem{tat}
S.~Tatikonda and S.~Mitter, ``The capacity of channels with feedback,''
  \emph{IEEE Transactions on Information Theory}, jan. 2009.

\bibitem{finitestatechannels}
H.~H. Permuter, T.~Weissman, and A.~Goldsmith, ``Finite state channels with
  time-invariant deterministic feedback,'' \emph{IEEE Transactions on
  Information Theory}, 2009.

\bibitem{PermuterWeissmanChenMACIT09}
H.~H. Permuter, T.~Weissman, and J.~Chen, ``Capacity region of the finite-state
  multiple access channel with and without feedback,'' \emph{IEEE Transactions
  on Information Theory}, vol.~55, pp. 2455--2477, 2009.

\bibitem{ShraderPermuter09CompoundIT}
B.~Shrader and H.~Permuter, ``Feedback capacity of the compound channel,''
  \emph{IEEE Transactions on Information Theory}, vol.~55, no.~8, pp. 3629
  --3644, 2009.

\bibitem{DaborahGoldsmith10BCfeedback}
R.~Dabora and A.~Goldsmith, ``The capacity region of the degraded finite-state
  broadcast channel,'' \emph{IEEE Transactions on Information Theory}, vol.~56,
  no.~4, pp. 1828 --1851, 2010.

\bibitem{iddoandhaim}
H.~H. Permuter and I.~Naiss, ``Extension of the blahut-arimoto algorithm for
  maximizing directed information,'' 2010.

\bibitem{Young}
S.~Yang, A.~Kavcic, and S.~Tatikonda, ``Feedback capacity of finite-state
  machine channels,'' \emph{IEEE Transactions on Information Theory}, 2005.

\bibitem{Gallager68}
R.~G. Gallager, \emph{Information theory and reliable communication}.\hskip 1em
  plus 0.5em minus 0.4em\relax New York: Wiley, 1968.

\bibitem{Arapos93_average_cose_survey}
A.~Arapostathis, V.~S. Borkar, E.~Fernandez-Gaucherand, M.~K. Ghosh, and
  S.~Marcus, ``Discrete time controlled {M}arkov processes with average cost
  criterion - a survey,'' \emph{SIAM Journal of Control and Optimization},
  vol.~31, no.~2, pp. 282--344, 1993.

\bibitem{enumerative}
T.~Cover, ``Enumerative source encoding,'' \emph{IEEE Transactions on
  Information Theory}, vol.~19, Jan 1973.

\bibitem{tofeedornottofeed}
H.~Asnani, H.~Permuter, and T.~Weissman, ``To feed or not to feed back,'' 2010.

\bibitem{new}
A.~Anastasopoulos, ``A sequential transmission scheme for unifilar finite-state
  channels with feedback based on posterior matching,'' in \emph{IEEE
  International Symposium on Information Theory}, 2012.

\bibitem{convex_optimization_boyd}
S.~Boyd and L.~Vandenberghe, \emph{Convex Optimization}.\hskip 1em plus 0.5em
  minus 0.4em\relax New York, USA: Cambridge University Press, 2004.

\end{thebibliography}

\end{document}